\documentclass[a4paper,USenglish,cleveref,thm-restate,nosubfigcap,noalgorithm2e,notab,nolineno]{lipics-v2021}

\hideLIPIcs  %

\nolinenumbers

\title{Revisiting ILP Models for Exact Crossing Minimization in Storyline Drawings} %

\author{Alexander Dobler}{TU Wien, Vienna, Austria}{adobler@ac.tuwien.ac.at}{0000-0002-0712-9726}{Vienna Science and Technology Fund (WWTF)  grant [10.47379/ICT19035]}

\author{Michael Jünger}{University of Cologne, Germany}{}{}{}

\author{Paul J.\ Jünger}{University of Bonn, Germany}{paul.j.juenger@gmail.com}{}{}

\author{Julian Meffert}{University of Bonn, Germany}{s6jumeff@uni-bonn.de}{0009-0008-9670-4569}{Deutsche Forschungsgemeinschaft (DFG, German Research Foundation) – 459420781}

\author{Petra Mutzel}{University of Bonn, Germany}{pmutzel@uni-bonn.de}{0000-0001-7621-971X}{partially funded by the Deutsche Forschungsgemeinschaft (DFG, German Research Foundation) – 459420781}

\author{Martin Nöllenburg}{TU Wien, Vienna, Austria}{noellenburg@ac.tuwien.ac.at}{0000-0003-0454-3937}{Vienna Science and Technology Fund (WWTF)  grant [10.47379/ICT19035]}

\authorrunning{A. Dobler, M. Jünger, P. J. Jünger, J. Meffert, P. Mutzel, M. Nöllenburg} %

\Copyright{Alexander Dobler, Michael Jünger, Paul J. Jünger, Julian Meffert, Petra Mutzel, Martin Nöllenburg}

\usepackage{hyperref,cite}
\usepackage[basic]{complexity}
\usepackage{todonotes}
\usepackage[inkscapelatex=false]{svg}
\usepackage{booktabs}
\usepackage{caption}
\usepackage{subcaption}
\usepackage{siunitx}
\usepackage{longtable,booktabs}

\usepackage{multicol}
\usepackage{siunitx}

\newcommand{\tm}{\operatorname{{time}}}
\newcommand{\charac}{\operatorname{{char}}}
\newcommand{\crossings}{\operatorname{{cr}}}
\newcommand{\assign}{\operatorname{{assign}}}
\newcommand{\lin}{\textcolor{lipicsGray}{\textbf{\sffamily LIN}}}
\newcommand{\qdr}{\textcolor{lipicsGray}{\textbf{\sffamily QDR}}}
\newcommand{\plo}{\textcolor{lipicsGray}{\textbf{\sffamily PLO}}}
\newcommand{\mc}{\textcolor{lipicsGray}{\textbf{\sffamily MC}}}
\newcommand{\qone}{\textcolor{lipicsGray}{\textbf{\sffamily Q1}}}
\newcommand{\qtwo}{\textcolor{lipicsGray}{\textbf{\sffamily Q2}}}
\newcommand{\qthree}{\textcolor{lipicsGray}{\textbf{\sffamily Q3}}}
\newcommand{\sbc}{\textcolor{lipicsGray}{\textbf{\sffamily SBC}}}
\newcommand{\init}{\textcolor{lipicsGray}{\textbf{\sffamily INIT}}}
\newcommand{\rnd}{\textcolor{lipicsGray}{\textbf{\sffamily RND}}}
\newtheorem{problem}{Problem}

\usepackage{mathtools}



\supplement{Source code and benchmark data: \href{https://doi.org/10.17605/OSF.IO/3BUA2}{\texttt{https://osf.io/3bua2/}}.}%

\begin{document}

\maketitle

\begin{abstract}
Storyline drawings are a popular visualization of interactions of a set of characters over time, e.g., to show participants of scenes in a book or movie. 
Characters are represented as $x$-monotone curves that converge vertically for interactions and diverge otherwise. 
Combinatorially, the task of computing storyline drawings reduces to finding a sequence of permutations of the character curves for the different time points, with the primary objective being crossing minimization of the induced character trajectories.
In this paper, we revisit exact integer linear programming (ILP) approaches for this \NP-hard problem. 
By enriching previous formulations with additional problem-specific insights and new heuristics, we obtain exact solutions for an extended new benchmark set of larger and more complex instances than had been used before. 
Our experiments show that 
our enriched formulations lead to better performing algorithms when compared to state-of-the–art modelling techniques. 
In particular, our best algorithms are on average 2.6--3.2 times faster than the state-of-the-art and succeed in solving complex instances that could not be solved before within the given time limit. Further, we show in an ablation study that our enrichment components contribute considerably to the performance of the new ILP formulation.

\end{abstract}

\section{Introduction} 
\begin{figure}[t]
    \centering
    \begin{subfigure}{\textwidth}
        \includegraphics[width=\textwidth]{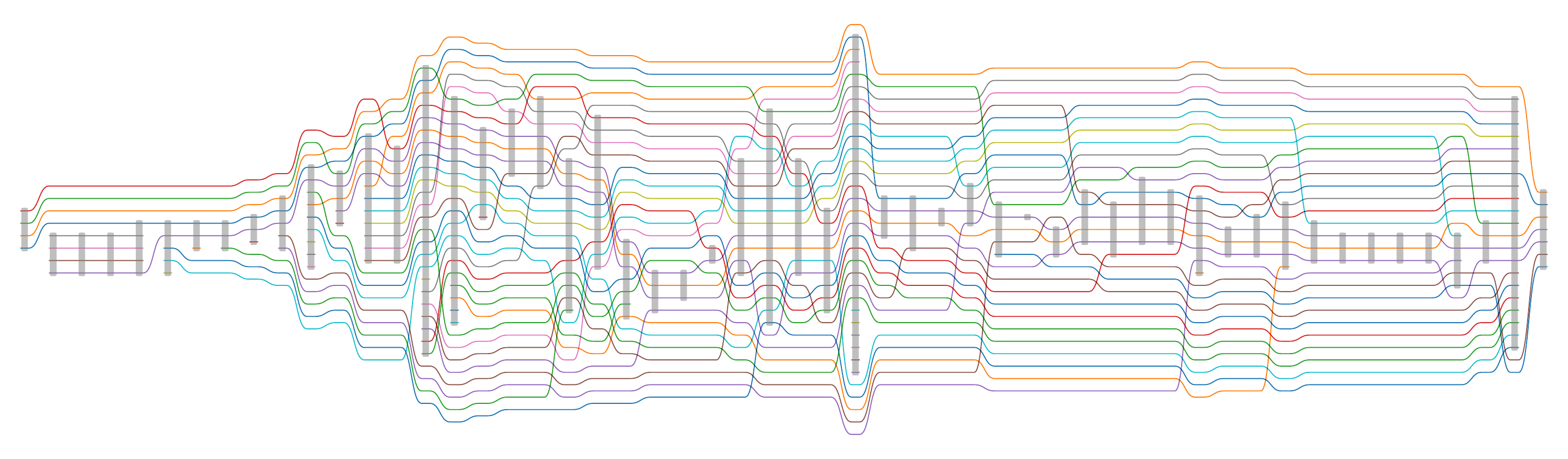}
        \caption{Solution with 374 crossings computed in \qty{0.15}{\second} by a greedy heuristic.}
        \label{fig:potter-input}
    \end{subfigure}
    \begin{subfigure}{\textwidth}
        \includegraphics[width=\textwidth]{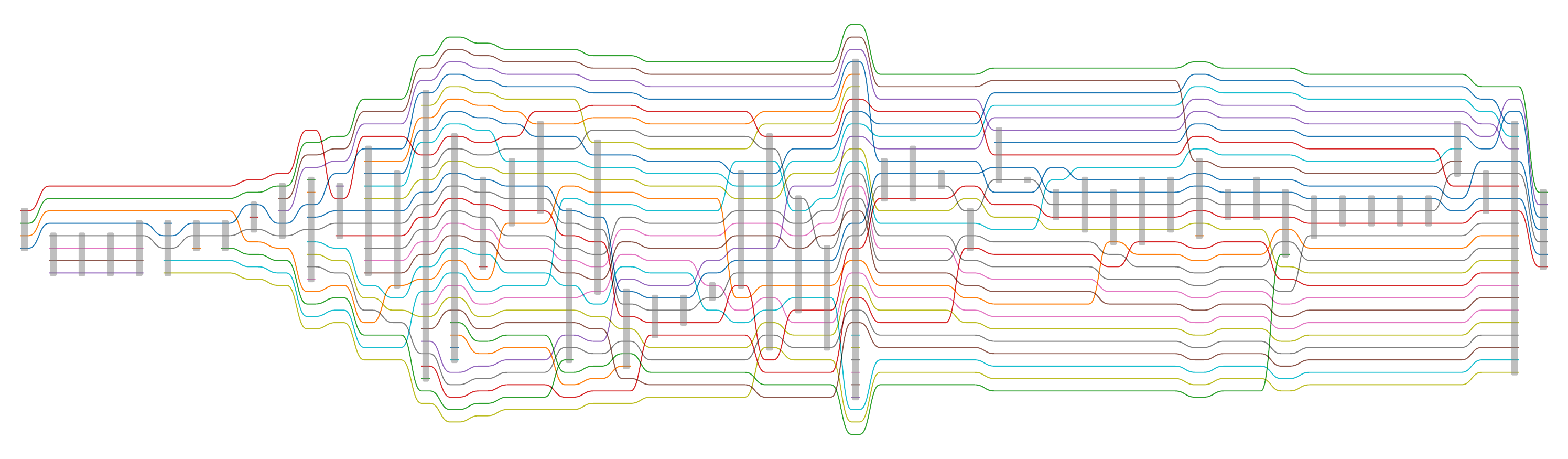}
        \caption{Drawing with the minimum number of 236 crossings computed in \qty{401.23}{\second} by our new ILP formulation.}
        \label{fig:potter-result}
    \end{subfigure}    
    \caption{Storylines for the first Harry Potter movie. Interactions are shown as vertical gray bars.}\label{fig:potter}
\end{figure}
Storyline drawings are a well-studied visualization style for complex event-based temporal interaction data and have been popularized by the xkcd comic ``Movie Narrative Charts'' in 2009~\cite{m-mnc-09}.
They show a set of characters, e.g., from the plot of a movie or book, and how they interact or co-occur in a sequence of events over time, e.g., by participating in the same scene or conversation of the evolving story.
A storyline drawing is a two-dimensional diagram, where the $x$-axis represents time and the $y$-dimension is used for the vertical grouping of characters according to their interaction sequence. The exact temporal distance of interactions is usually not depicted, only their order.
Each character is represented as an $x$-monotone curve, and interactions are represented by vertically grouping the curves of the participating characters at the $x$-coordinate corresponding to the interaction time.
Characters that are not participating in an interaction at any specific point in time are vertically separated from each other.
\Cref{fig:potter} shows an example of a storyline drawing.
Due to their popularity and the intuitive data encoding, they are well suited for visual storytelling and have since been used as visual metaphors for representing a variety of different event-based data sets beyond the original book and movie plots~\cite{m-mnc-09,tm-dcosv-12}, e.g., for software projects~\cite{thm-efgsvfsd-15,om-ses-10}, newspaper articles~\cite{axpab-hiudh-22}, political debates on social media~\cite{lwwll-stes-13}, visual summaries of meeting contents~\cite{sbbzzm-mvnarmcc-18}, scientific collaborations~\cite{h-sfwk-22}, sports commentary~\cite{wszlcl-evrwteed-24}, and gameplay data~\cite{wwd-vsegusvswll-23}.

There is one main degree of freedom when computing and optimizing storyline drawings: the (vertical) linear order and positioning of the characters at each discrete time steps. 
The only hard constraint is that all characters participating in the same interaction must be consecutive as a group.
This degree of freedom can thus be used to minimize the number of crossings between character curves, their wiggles (i.e. the amount of vertical movement of character lines in the visualization), and any excessive white space in the diagram, which are the three major objectives that have been identified for computing storyline drawings~\cite{tm-dcosv-12,thm-efgsvfsd-15}.
While the number of crossings is determined purely combinatorially by the sequence of character permutations, wiggles and white space depend on the actual $y$-coordinates assigned to each character curve at each point in time.

In this paper, our interest is the combinatorial crossing minimization problem. 
It is the primary objective in practical storyline optimization pipelines~\cite{lwwll-stes-13}, where it forms the input to the subsequent steps of reducing wiggles and white space while maintaining the character order.
Additionally, crossing minimization is one of the most fundamental graph drawing problems~\cite{s-cng-18,bett-gd-99} and it is well known that graph drawings with fewer crossings increase readability~\cite{p-wagehu-97}. 
Unfortunately, crossing minimization in storyline drawings is an \NP-hard problem~\cite{knpss-mcsv-15,gj-cnn-83} and hence practical approaches for storyline visualization usually resort to %
heuristics, even though they cannot guarantee optimal solutions.  

We consider the crossing minimization problem from the opposite side and revisit exact integer linear programming (ILP) approaches~\cite{GronemannJLM16} for computing provably optimal solutions. Such approaches often lead to practical exact algorithms.
Our goal is to improve on the runtime performance of such exact methods by enriching the models with both new problem-specific insights and better heuristics. %
Faster exact algorithms for crossing minimization in storyline drawings are practically relevant for two reasons: firstly, solving moderately sized instances to optimality within a few seconds provides a strictly better alternative to commonly used suboptimal heuristics, and secondly, knowing optimal solutions for a large set of representative benchmark instances (even if their computations take several minutes or up to a few hours) is a prerequisite for any thorough experimental study on the performance of non-exact crossing minimization heuristics and for generating crossing-minimum stimuli in user experiments.

\subparagraph{Related Work.}
Tanahashi and Ma~\cite{tm-dcosv-12} introduced storyline drawings as an information visualization problem, provided the first visual encoding model, and defined the above-mentioned optimization criteria (crossings, wiggles, white space). 
They suggested a genetic algorithm to compute storyline drawings. 
Ogawa and Ma~\cite{om-ses-10} used a greedy algorithm to compute storylines to depict software evolution. Due to slow computation times of previous methods, Liu et al.~\cite{lwwll-stes-13} split the layout process into a pipeline of several subproblems ordered by importance, the first one being crossing minimization. 
They solved the character line ordering by an iterated application of a constrained barycenter algorithm, a classic heuristic for multi-layer crossing minimization~\cite{DBLP:journals/tsmc/SugiyamaTT81}. 
Their results were obtained in less than a second and had fewer crossings than those computed by the genetic algorithm~\cite{tm-dcosv-12}, which took more than a day to compute on some of the same instances. 
Tanahashi et al.~\cite{thm-efgsvfsd-15} enhanced previous methods to take into account streaming data and apply a simple sequential left-to-right sorting heuristic. %
Recent practical works on storyline drawings focus on other aspects, such as an interactive editor~\cite{trlcyw-iechs-18} or a mixed-initiative system including a reinforcement learning AI component~\cite{tlwlkk-pcesvurl-20}; both these systems apply a two-layer crossing minimization heuristic~\cite{f-fshctcr-04}.  %

Several authors focused on the combinatorial crossing minimization problem and its complexity. 
Kostitsyna et al.~\cite{knpss-mcsv-15} observed that the \NP-hardness of the problem follows from a similar bipartite crossing minimization problem~\cite{gj-cnn-83} and proved fixed-parameter tractability when the number of characters is bounded by a parameter $k$. 
Gronemann et al.~\cite{GronemannJLM16} were the first to model the problem as a special type of tree-constrained multi-layer crossing minimization problem. They implemented an exact branch-and-cut approach that exploits the equivalence of the quadratic unconstrained 0/1-optimization problem with the maximum cut problem in a graph.
They managed to solve many instances with up to 20 characters and 50 interactions optimally within a few seconds.
Van Dijk et al.~\cite{dlmw-csvwbc-17,dfflmr-bcsv-17} proposed block-crossing minimization in storyline drawings, which counts grid-like blocks of crossings rather than individual crossings. 
They showed \NP-hardness and proposed greedy heuristics, a fixed-parameter tractable algorithm, and an approximation algorithm. 
In a follow-up work, van Dijk et al.~\cite{dlmw-csvwbc-17} implemented and experimentally evaluated an exact approach for the block crossing minimization problem using SAT solving.
A different variation of storylines was studied by Di Giacomo et al.~\cite{di2020storyline}, who considered ubiquitous characters as $x$-monotone trees with multiple branches, enabling characters to participate in multiple simultaneous interactions; they solved the crossing minimization aspect using an adaptation of the previous SAT model~\cite{dlmw-csvwbc-17}.
Dobler et al.~\cite{DoblerNSVW23} consider time interval storylines, where additionally to the order of characters, the order of time steps in so-called time-intervals can be permuted.

The problem is also similar to crossing minimization in layered graph drawing, which was introduced by Sugiyama et al.~\cite{DBLP:journals/tsmc/SugiyamaTT81}. The problem is to draw a graph with its vertices on multiple parallel lines while minimizing crossings. A notable difference to storyline crossing minimization is that vertices can have arbitrary degree and that edges can span more than one layer. For a survey of algorithms and techniques in layered graph drawing, we refer to Healy and Nikolov~\cite{DBLP:reference/crc/HealyN13}.

\subparagraph{Contributions.}
\sloppy
The contributions of this paper are the following: 
\begin{itemize}
    \item We identify structural properties of storyline drawings and prove that there exist crossing-minimum drawings satisfying them, reducing the search space of feasible solutions.
    \item We propose a new ILP formulation exploiting these structural insights in order to (i) significantly reduce the number of required constraints and (ii) apply symmetry breaking constraints to strengthen the ILP model. %
    \item We introduce several new heuristics that support the exact solver, either as initial heuristics to improve branch-and-bound pruning or for deriving integral solutions from fractional ones during the incremental ILP solving process. 
    \item We have compiled a new benchmark set of storyline instances, including those of earlier studies, as well as several challenging new ones.
    \item We have conducted a detailed experimental evaluation of our new ILP model using the above benchmark set. We compare its ability to solve instances with state-of-the-art ILP models. Moreover, in an ablation study, we show that our further enhancements (e.g., adding symmetry breaking constraints and novel heuristics) contributes considerably to the performance of both the new and several state-of-the-art ILP formulations.
    \item We show that our ILP models are able to solve previously unsolved instances from the literature and obtain a speedup of 2.6--3.2 compared to the state of the art.
\end{itemize}
Data sets, source code, evaluation, and a visualization software are available on \href{https://doi.org/10.17605/OSF.IO/3BUA2}{\texttt{https://osf.io/3bua2/}}.

\section{Preliminaries}
\subparagraph{Permutations.}
Given a set $X=\{x_1,\ldots,x_n\}$, a permutation $\pi$ is a linear order of its elements, or equivalently, a bijective mapping from $\{1,2,\dots,|X|\}$ to $X$. For $x,x'\in X$ we write $x\prec_{\pi}x'$ if $x$ comes before $x'$ in $\pi$.
For $Y\subseteq X$, $\pi[Y]$ is the permutation $\pi$ restricted to $Y$, formally, for $y,y'\in Y$,  $y\prec_{\pi[Y]}y'$ if and only if $y\prec_{\pi}y'$. For two permutations $\pi,\phi$ of two sets $X$ and $Y$ with $X\cap Y=\emptyset$, we denote by $\pi\star \phi$ their concatenation. Given two permutations $\pi,\pi'$ of the same set $X$, the \emph{inversions} between $\pi$ and $\pi'$ is the number of pairs $x,x'\in X$ such that $\pi^{-1}(x)<\pi^{-1}(x')$ and $\pi'^{-1}(x)>\pi^{-1}(x')$.

\subparagraph{Problem input.}
\begin{figure}
    \centering
    \includegraphics{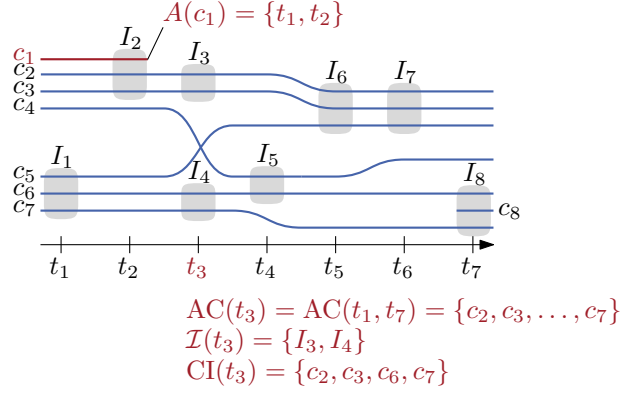}
    \caption{Illustration of important notation throughout this paper with the aid of a storyline drawing depicting interactions {$I_1$}--{$I_8$} of the characters {$c_1$}--{$c_8$} over the time steps {$t_1$}--{$t_7$}}
    \label{fig:concepts}
\end{figure}
A \emph{storyline instance} consists of a 4-tuple $(T,\mathcal{C},\mathcal{I}, A)$ where $T=\{t_1,t_2,\dots,t_\ell\}$ is a set of totally ordered \emph{time steps} (or \emph{layers}), $\mathcal{C}=\{c_1,c_2,\dots,c_n\}$ is a set of \emph{characters}, and $\mathcal{I}=\{I_1,I_2,\dots,I_m\}$ is a set of interactions. Each interaction $I\in \mathcal{I}$ has a corresponding time step $\tm(I) \in T$ and consists of a set of characters $\charac(I) \subseteq \mathcal{C}$. Further, $A$ maps each character $c\in \mathcal{C}$ to a consecutive set of time steps, i.e., $A(c)=\{t_i,t_{i+1},\dots, t_j\}$ for $1\le i\le j\le \ell$. We say that character $c$ is \emph{active} at the time steps in $A(c)$, it \emph{starts} at $t_i$ and \emph{ends} at $t_j$.
We define $\operatorname{AC}(t)$ for $t\in T$ as the set of all characters $c\in \mathcal{C}$ active at time $t$, i.e., %
$\operatorname{AC}(t)=\{c\in \mathcal{C}\mid t\in A(c)\}$.
Clearly, for each interaction $I\in\mathcal{I}$, $\charac(I)\subseteq \operatorname{AC}(\tm(I))$ must hold.
Next, we define the set of all characters active in the time interval $[t_i,t_j]$ ($1\le i\le j\le \ell$) as $\operatorname{AC}(t_i,t_j)=\operatorname{AC}(t_i)\cap \operatorname{AC}(t_{i+1})\cap\dots\cap \operatorname{AC}(t_j)$. 
For a time step $t\in T$ we define the set of interactions at $t$ as $\mathcal{I}(t)=\{I\in\mathcal{I}\mid \tm(I)=t\}$ and its corresponding set of characters as $\operatorname{CI}(t)=\bigcup_{I\in \mathcal{I}(t)}\charac(I)$. Without loss of generality, for the interactions at time step $t$ we assume that $|\mathcal{I}(t)|\ne 0$ and that the sets of characters of the interactions $\mathcal{I}(t)$ are pairwise disjoint. This is a reasonable assumption as characters usually participate in at most one interaction at any given time, e.g.\ in movies. Important notation is also illustrated in \cref{fig:concepts}.
\subparagraph{Problem output.}
Solutions to storyline instances $(T,\mathcal{C},\mathcal{I}, A)$ consist of a sequence of $\ell$ permutations $S=(\pi_1,\pi_2,\dots,\pi_\ell)$ such that $\pi_i$ is a permutation of $\operatorname{AC}(t_i)$ for all $i=1,\ldots,\ell$ satisfying the condition that the set of characters of each interaction $I\in \mathcal{I}(t_i)$ 
appears consecutively.
We call $S$ a \emph{storyline solution} or \emph{drawing}.

The number of crossings $\text{cr}(\pi_i, \pi_{i+1})$ between two consecutive permutations $\pi_i$ and $\pi_{i+1}$ is defined as the number of inversions of the two permutations $\pi_i[C]$ and $\pi_{i+1}[C]$, where $C=\operatorname{AC}(t_i)\cap \operatorname{AC}(t_{i+1})$. The number of crossings in a storyline solution is $\sum_{i=1}^{\ell-1}\text{cr}(\pi_i,\pi_{i+1})$. The problem addressed in this paper is the following.
\begin{problem}[Storyline Problem]\label{prob:storyline}
    Given a storyline instance $(T,\mathcal{C},\mathcal{I},A)$, find a storyline drawing $S$ with the minimum number of crossings.
\end{problem}

\section{Standard Models for the Storyline Problem}\label{section:models} %
The most natural ILP formulation to solve \cref{prob:storyline} has a quadratic objective function and is based on the linear ordering model, which uses binary variables in order to encode the linear ordering at each time step. The number of crossings between two subsequent time steps is then given by the number of inversions of the two permutations.

From now on, we assume that characters $c_u,c_v,c_w$ are pairwise different, even if we write for example $c_u,c_v\in C$ for some set $C$ of characters.

\subsubsection*{Quadratic Model (QDR)}
For each time step $t_i,i=1,2,\dots,\ell$ and each tuple of characters $c_u,c_v\in \operatorname{AC}(t_i)$ we introduce a binary \emph{ordering variable} $x_{i,u,v}$ which is equal to 1 if and only if $c_u\prec_{\pi_i}c_v$.
The quadratic model QDR is given as follows:

\begin{align}
    \min \quad & \sum_{i=1}^{\ell-1}\sum_{c_u,c_v\in \operatorname{AC}(t_i,t_{i+1})}x_{i,u,v}x_{i+1,v,u}   &   &  \tag{QDR}\label{QDR-ILP:c}\\
  x_{i,u,v} = 1- x_{i,v,u} &\qquad \textrm{for all } i=1,\ldots,\ell; c_u,c_v\in \operatorname{AC}(t_i) \textrm{ with } u<v \tag{EQ} \label{EQ-Con}\\
        x_{i,u,v}+x_{i,v,w}+x_{i,w,u}\le 2 &\qquad \textrm{for all }  i=1,\ldots,\ell; c_u,c_v,c_w\in \operatorname{AC}(t_i) %
        \tag{LOP} \label{LOP-Con}\\
   x_{i,u,w}=x_{i,v,w} &\qquad \textrm{for all }  i=1,\ldots,\ell; I\in {\cal I}(t_i); \tag{TREE} \label{TREE-Con} \\
   &\qquad  c_u,c_v\in \charac(I), u<v; c_w\in \operatorname{AC}(t_i)\setminus \charac(I)\notag \\
   x_{i,u,v}\in\{0,1\} &\qquad \textrm{for all } i=1,\ldots,\ell; c_u,c_v\in \operatorname{AC}(t_i), \tag{BIN} \label{BIN-Con}
\end{align}

 The character curves for $c_u$ and $c_v$ cross between the two layers $t_i$ and $t_{i+1}$ if and only if
 one of the terms $x_{i,u,v}x_{i+1,v,u}$ and $x_{i,v,u}x_{i+1,u,v}$ equals $1$.
The \eqref{LOP-Con} and \eqref{EQ-Con} constraints ensure transitivity of the set of characters $\operatorname{AC}(t_i)$ present at time step $t_i$ and guarantee that they define a total order.
For all interactions $I\in {\cal I}(t_i)$ %
the \eqref{TREE-Con} constraints ensure that characters from $I$ appear consecutively at the respective time step $t_i$. %

\subsubsection*{Linearized Model (LIN)}
The standard linearisation of quadratic integer programs introduces additional variables $y_{i,u,v}$ that substitute the quadratic terms $x_{i,u,v}x_{i+1,v,u}$ for all $t_i$, $i=1,2,\dots,\ell-1$ and each tuple of characters $c_u,c_v\in \operatorname{AC}(t_i,t_{i+1})$ in the objective function. %
In order to link the new variables with the ordering variables, we introduce the following constraints:
\begin{align*}
    & \qquad\qquad y_{i,u,v}\ge x_{i,u,v}-x_{i+1,u,v} & \quad \textrm{for all } i=1,\dots,\ell; c_u,c_v\in \operatorname{AC}(t_i,t_{i+1}) \tag{CR} \label{CR-Con}
\end{align*}
Obviously, the variable $y_{i,u,v}$ is forced to $1$, if the character $c_u$ is before $c_v$ at time step $t_i$ in the solution represented by the $y$-variables, and the order of both characters is reversed at time step $t_{i+1}$.
The linearised model \eqref{LIN-ILP} is given as follows.
\begin{align*} 
\qquad\qquad  \min & \quad \sum_{i=1}^{\ell-1}\sum_{c_u,c_v\in \operatorname{AC}(t_i,t_{i+1})}y_{i,u,v}  & & \tag{LIN} \label{LIN-ILP}\\
 & \quad y_{i,u,v},x_{i,u,v} \text{ satisfy  (\ref{BIN-Con}), (\ref{EQ-Con}), (\ref{LOP-Con}), (\ref{TREE-Con}), and (\ref{CR-Con})}  
\end{align*}

\subsubsection*{Max-Cut Model (CUT)}

Gronemann et al.~\cite{GronemannJLM16} have suggested a formulation based on the transformation of the problem into a quadratic unconstrained binary program with additional (TREE) constraints, which is then solved using a maximum cut approach. Here, we omit the detour via the quadratic binary program and directly provide the corresponding maximum cut formulation. 
Starting with a feasible storyline drawing $\hat{S}=(\hat{\pi}_1,\dots,\hat{\pi}_\ell)$, we define the graph $G_M=(V_M,E_M)$: 
The vertex set $V_M$ is given by a vertex $v^*$ and the union of the sets $V_i$ ($i=1,\ldots,\ell)$, where $V_i$ has a vertex $c_{uv}^i$ for each pair $c_u,c_v\in \operatorname{AC}(t_i)$ with $c_u\prec_{\hat{\pi}_i}c_v$.%

We introduce an edge between the vertices $c_{uv}^i$ and $c_{pq}^{i+1}$ if the corresponding characters coincide.
In the case that $c_u=c_p$ and $c_v=c_q$, the (type-1) edge $e=e_{uv}^i$ gets a weight of $w_e=-1$, and in the case that $c_u=c_q$ and $c_v=c_p$, the (type-2) edge $e=e_{uv}^i$ gets a weight of $w_e=1$.
We define the constant $K$ as the number of edges of type (2). 
The intention is the following: 
An edge of type (1) results in a crossing if and only if it is in the cut, and an edge of type (2) results in a crossing if and only if it is not in the cut.
This construction allows for associating the maximum cut objective function values $W$ to corresponding crossing numbers $K-W$. 
In particular, $W=0$ for the empty cut corresponds to the number of crossings $K$ in $\hat{S}$.
In order to guarantee that the characters of an interaction appear consecutively, we introduce type-3 edges with weight 0 from the additional vertex $v^*$ to every vertex in $V_i$ for all $i=1,\ldots,\ell$, and add the additional constraints~(\ref{TREE-MC}).
We introduce a binary variable $z_e$ for every edge $e\in E_M$ in the graph, which is 1 if and only if the edge is contained in the computed cut.

The following model guarantees that every optimal solution corresponds to a constrained maximum cut in the graph $G_M$ that provides the optimal solution to the storyline problem. 
The constraints (\ref{CYC-MC}) capture the fact that any intersection of a cut and a cycle in a graph has even cardinality.
The correctness is provided in~\cite{GronemannJLM16}, see also \cite{DBLP:journals/dm/Simone90,DBLP:journals/informs/BuchheimWZ10}.
\begin{align*}
\qquad \max & \quad \sum_{e\in E_M}w_ez_e \tag{CUT}\label{cut-formulation}\\
 \sum_{e\in F}z_e-\sum_{e\in C\setminus F}z_e\le\vert F\vert -1  &\quad \mbox{for all cycles $C\subseteq E_M$}, F\subseteq C, \vert F \vert \textrm{ odd}  \label{CYC-MC}  \tag{CYC} \\   %
        0\le z_{(v^*,c_{uv}^i)}+z_{(v^*,c_{vw}^i)}-z_{(v^*,c_{uw}^i)}\le 1 &\quad \textrm{for all }  i=1,\ldots,\ell; c_{uv}^i,c_{vw}^i,c_{uw}^i\in V_i
        \tag{LOPC} \label{LOPC-Con}\\
        &\quad \textrm{with } c_u\prec_{\hat{\pi}_i}c_v\prec_{\hat{\pi}_i}c_w\\
   \begin{rcases*} z_{(v^*,c_{uw}^i)}=z_{(v^*,c_{vw}^i)}\text{ if }c_u,c_v\prec_{\hat{\pi}_i}c_w\\
       z_{(v^*,c_{wu}^i)}=z_{(v^*,c_{wv}^i)}\text{ if }c_u,c_v\succ_{\hat{\pi}_i}c_w
       \end{rcases*}&\quad\begin{aligned}
       &\textrm{for all }  i=1,\ldots,\ell; I\in {\cal I}(t_i);\\&c_u,c_v\in \charac(I); c_w\in \operatorname{AC}(t_i)\setminus \charac(I)
   \end{aligned}\tag{TRC}\label{TREE-MC} \\
 z_e\in\{0,1\} &\quad \mbox{for all $e\in E_M$}\tag{BIC}
\end{align*}

\section{Structural Properties of Storyline Solutions}\label{section:structural}
In this section, we identify structural properties of storyline solutions that will help us to optimize the models proposed in \cref{section:newformulation}, and that guide the exact optimization process. 

For the results in this section, we introduce some definitions and observe some properties of the function $\crossings$.
Let $\pi$ be a permutation of the set $X$, and $\phi$ be a permutation of the set $Y\subseteq X$.
We define the permutation $\psi=\assign(\pi,\phi)$ as the permutation of the set $X$ such that
\[\psi(i)=\begin{cases}\pi(i)\qquad\text{if }\pi(i)\in X\setminus Y\\
                        \phi(j)\qquad\text{if }c=\pi(i)\in Y\text{ with }j=\pi[Y]^{-1}(c)\end{cases}\]
In particular, all elements in $X\setminus Y$ are ordered according to $\pi$, and all elements in $Y$ switch their position so that they are ordered according to $\phi$.
Notice that $\assign(\pi,\phi)[Y]=\phi$.
\begin{example}
    Let $\pi=(c,e,b,d,f,a,g)$ and $\phi=(a,b,c)$. Then $\assign(\pi,\psi)=(a,e,b,d,f,c,g)$.
\end{example}
Consider three consecutive time steps $t_i,t_{i+1},t_{i+2}$ ($1\le i\le\ell-2)$, we observe a special type of triangle inequality.
\begin{observation}\label{obs:triangle}
    Let $C\subseteq \operatorname{AC}(t_i,t_{i+2})$. For any feasible solution $S$ we have
    \[\crossings(\pi_i[C],\pi_{i+1}[C])+\crossings(\pi_{i+1}[C],\pi_{i+2}[C])\ge \crossings(\pi_i[C],\pi_{i+2}[C]).\]
\end{observation}
Next, consider two permutations $\pi_i,\pi_{i+1}$ of a storyline solution such that $C=\operatorname{AC}(t_i)\cap \operatorname{AC}(t_{i+1})$. 
Consider two sets $X,Y\subseteq C$. We define $\crossings(\pi_i,\pi_{i+1},X,Y)$ as the amount of $(c,c')\in X\times Y$ such that either $c\prec_{\pi_i}c'$ and $c'\prec_{\pi_{i+1}}c$, or $c'\prec_{\pi_i}c$ and $c\prec_{\pi_{i+1}}c'$. I.e., these are the crossings between time steps $t_i$ and $t_{i+1}$ with one character in $X$ and one character in $Y$. We set $\crossings(\pi_i,\pi_{i+1},X)=\crossings(\pi_i,\pi_{i+1},X,X)$.
\begin{observation}\label{obs:inducedcrossings}
    It holds that $\crossings(\pi_i,\pi_{i+1},X)=\crossings(\pi_i[X],\pi_{i+1}[X])$.
\end{observation}
\begin{observation}\label{obs:crossingdecomposition}
Let $X,Y$ be sets such that $X\dot{\cup}Y=C\subseteq \operatorname{AC}(t_i,t_{i+1})$. It holds that
\[\crossings(\pi_i,\pi_{i+1}, C)=\crossings(\pi_i,\pi_{i+1},X)+\crossings(\pi_i,\pi_{i+1},Y)+\crossings(\pi_i,\pi_{i+1},X,Y).\]
\end{observation}

Now, we define two properties of storyline drawings. \cref{definition:consistent1} captures that the relative order of characters in an interaction can be propagated backwards.
\begin{definition}[Type-1 consistency]\label{definition:consistent1}
    Let $S=(\pi_1,\pi_2,\dots,\pi_\ell)$ be a solution to a storyline instance $(T,\mathcal{C},\mathcal{I}, A)$. Let $I\in \mathcal{I}$, $t_i=\tm(I)$ and $C=\charac(I)$. Let $1< j(I)\le i$ be the index of the earliest time step $t_{j(j)}$ such that $C\subseteq \operatorname{AC}(t_{j(I)},t_i)$ and
    \[\forall k\in \{j(I)+1,\dots, i\}:\operatorname{CI}(t_k)\cap C=\emptyset\lor \exists I\in\mathcal{I}(t_k):C\subseteq \charac(I).\]
     We say that $S$ is $I$-consistent if
     \[\forall k\in \{j(I),j(I)+1,\dots,i\}:\pi_k[C]=\pi_i[C].\]
     Further, we say that $S$ is \emph{type-1-consistent} if it
     is $I$-consistent for all $I\in \cal{I}$.
\end{definition}

\cref{definition:consistent2} defines the property that between suitable pairs of interactions with the same set of characters, these characters are kept together between the two time steps. Note that this is not implied by type-1 consistency.
\begin{definition}[Type-2 consistency]\label{definition:consistent2}
    Let $S=(\pi_1,\pi_2,\dots,\pi_\ell)$ be a solution to a storyline instance $(T,\mathcal{C},\mathcal{I}, A)$. Consider two interactions $I_1,I_2\in \mathcal{I}$ such that 
    \begin{itemize}
        \item $\charac(I_1)=\charac(I_2)=C$,
        \item $i=\tm(I_1)<\tm(I_2)=j$, and
        \item $\forall k\in\mathbb{N}: i<k<j\Rightarrow[\operatorname{CI}(t_k)\cap C=\emptyset\lor \exists I_3\in \mathcal{I}(t_k):C\subseteq \charac(I_3)]$.
    \end{itemize}
    We say that $S$ is $(I_1,I_2)$-consistent if
    \[\forall i<k<j:\exists \pi^a,\pi^b:\pi_k=\pi^a\star \pi_i[C]\star \pi^b.\]
    Further, we say that $S$ is type-2-consistent if it is $(I_1,I_2)$-consistent for all such pairs $(I_1,I_2)$.
\end{definition}

The following lemma shows that we can achieve type-1 consistency for storyline drawings without increasing the number of crossings. Essentially, if a storyline solution is not type-1 consistent for an interaction $I$, we can propagate the relative order of characters in that interaction forward from the time step $t_{j(I)}$ from \cref{definition:consistent1}.
\begin{restatable}{lemma}{typeonelemma}
    \label{lemma:type1}
    Let $(T,\mathcal{C},\mathcal{I}, A)$ be an instance with a solution $S$. We can construct from $S$ a type-1-consistent solution $S'$ such that $\crossings(S')\le \crossings(S)$. If $S$ is type-2-consistent, so is $S'$.
\end{restatable}
\begin{proof}
    Let $S=(\pi_1,\pi_2,\dots,\pi_\ell)$ be a solution such that there exists $I\in\mathcal{I}$ such that $S$ is not $I$-consistent.
    We choose such an interaction $I$ that is as ``late as possible'', i.e.\ $I\in \mathcal{I}(t_i)$ with $i$ maximal.
    We construct $S'=(\pi_1',\pi_2',\dots,\pi_\ell')$ such that
    \begin{enumerate}[(1)]
        \item $\crossings(S')\le \crossings(S)$, 
        \item $S'$ is $I$-consistent,
        \item for each $I'\in\mathcal{I}$ such that $S$ is $I'$-consistent and $\tm(I')\ge \tm(I)$, $S'$ is also $I'$-consistent.
    \end{enumerate}
    Let $t_i=\tm(I)$ and $C=\charac(I)$ and define $j(I)<i$ as in \cref{definition:consistent1}.
    For all $1\le k\le \ell$ such that $k\le j(I)$ or $k>i$ we set $\pi_k'=\pi_k$.
    For $k\in \{j(I)+1,j(I)+2,\dots,i\}$ we set $\pi_k'=\assign(\pi_k,\pi_{j(I)}[C])$.
    It is clear that the new solution $S'$ is $I$-consistent and is still a solution to the storyline instance (characters involved in interactions are still consecutive in the respective time steps).
    
    We first argue that $\crossings(S')\le \crossings(S)$.
    First, it is clear that $\crossings(\pi'_k,\pi'_{k+1})=\crossings(\pi_k,\pi_{k+1})$ for $k\le j(I)-1$ and $k>i$ as the involved permutations did not change.
    It remains to show that %
    \begin{equation*}
        \sum_{k=j(I)}^i\crossings(\pi_k',\pi_{k+1}')\le \sum_{k=j(I)}^i\crossings(\pi_k,\pi_{k+1}).
    \end{equation*}
    We also have that 
    \begin{align*}
        \sum_{k=j(I)}^i\crossings(\pi_k',\pi_{k+1}')=&\sum_{k=j(I)}^i\crossings(\pi_k',\pi_{k+1}',C)+\\
            &\sum_{k=j(I)}^i\crossings(\pi_k',\pi_{k+1}',\operatorname{AC}(t_k,t_{k+1})\setminus C)+\\
            &\sum_{k=j(I)}^i\crossings(\pi_k',\pi_{k+1}',C,\operatorname{AC}(t_k,t_{k+1})\setminus C),
    \end{align*}
    and 
    \begin{align*}
        \sum_{k=j(I)}^i\crossings(\pi_k,\pi_{k+1})=&\sum_{k=j(I)}^i\crossings(\pi_k,\pi_{k+1},C)+\\
            &\sum_{k=j(I)}^i\crossings(\pi_k,\pi_{k+1},\operatorname{AC}(t_k,t_{k+1})\setminus C)+\\
            &\sum_{k=j(I)}^i\crossings(\pi_k,\pi_{k+1},C,\operatorname{AC}(t_k,t_{k+1})\setminus C).
    \end{align*}
    Thus it is enough to show that the inequality holds for the respective terms separately.
    \begin{itemize}
        \item $\sum_{k=j(I)}^i\crossings(\pi_k',\pi_{k+1}',C)\le \sum_{k=j(I)}^i\crossings(\pi_k,\pi_{k+1},C)$: We have
        \begin{align*}
            \sum_{k=j(I)}^i\crossings(\pi_k,\pi_{k+1},C)&=\sum_{k=j(I)}^i\crossings(\pi_k[C],\pi_{k+1}[C])\\
            &\ge \sum_{k=j(I)}^i\crossings(\pi_k[C\cap \operatorname{AC}(t_{i+1})],\pi_{k+1}[C\cap \operatorname{AC}(t_{i+1})])\\
            &\ge \crossings(\pi_{j(I)}[C\cap \operatorname{AC}(t_{i+1})], \pi_{i+1}[C\cap \operatorname{AC}(t_{i+1})])\\
            &=\crossings(\pi'_{i}[C\cap \operatorname{AC}(t_{i+1})], \pi_{i+1}'[C\cap \operatorname{AC}(t_{i+1})])\\
            &=\sum_{k=j(I)}^i\crossings(\pi_k',\pi_{k+1}',C)
        \end{align*}
        The first inequality holds because $C\cap \operatorname{AC}(t_{i+1})$ is a subset of $C$. The second inequality holds because of \cref{obs:triangle}. The last equality holds because the only crossings between characters from $C$ from $t_{j(I)}$ to $t_{i+1}$ appear between $t_{i}$ and $t_{i+1}$ as the relative order of characters from $C$ in $S'$ is the same from $j(I)$ to $i$.
        
        \item $\sum_{k=j(I)}^i\crossings(\pi_k',\pi_{k+1}',\operatorname{AC}(t_k,t_{k+1})\setminus C)\le \sum_{k=j(I)}^i\crossings(\pi_k,\pi_{k+1},\operatorname{AC}(t_k,t_{k+1})\setminus C)$: Let $k\in \{j(I), j(I)+1,\dots,i\}$ and let $C_a=\operatorname{AC}(t_k,t_{k+1})\setminus C$. We have that
        \begin{align*}
            \crossings(\pi_k',\pi_{k+1}',\operatorname{AC}(t_k,t_{k+1})\setminus C)&=\crossings(\pi_k'[C_a],\pi_{k+1}'[C_a])\\
            &=\crossings(\pi_k[C_a],\pi_{k+1}[C_a])\\
            &=\crossings(\pi_k,\pi_{k+1},\operatorname{AC}(t_k,t_{k+1})\setminus C).
        \end{align*}
        The key is that $\pi_k[C_a]=\pi_k'[C_a]$ and $\pi_{k+1}[C_a]=\pi_{k+1}'[C_a]$, as we did not change the relative order of non-interaction characters, that is, characters that are not in $C$.
        As the equality holds for each term, it also holds for the sum.
        \item $\sum_{k=j(I)}^i\crossings(\pi_k',\pi_{k+1}',C, \operatorname{AC}(t_k,t_{k+1})\setminus C)\le \sum_{k=j(I)}^i\crossings(\pi_k,\pi_{k+1},C, \operatorname{AC}(t_k,t_{k+1})\setminus C)$: We again argue for each $k\in \{j(I),j(I)+1,\dots,i\}$ separately and consider two cases.
        
        If $k=i$ then the relative order of character pairs $(c,c')\in C\times (\operatorname{AC}(t_k,t_{k+1})\setminus C)$ is the same for $\pi_i$ and $\pi_i'$ as $C$ must be consecutive in $t_i$, and $\pi[\operatorname{AC}(t_k,t_{k+1})\setminus C] = \pi'[\operatorname{AC}(t_k,t_{k+1})\setminus C]$. The relative order is also the same in $\pi_{i+1}'$ and $\pi_{i+1}$ as $\pi_{i+1}'=\pi_{i+1}$. Thus $\crossings(\pi_k,\pi_{k+1},C,\operatorname{AC}(t_k,t_{k+1)})\setminus C)=\crossings(\pi_k',\pi_{k+1}',C,\operatorname{AC}(t_k,t_{k+1})\setminus C)$ holds.

        In the remaining case, we have that $k\in \{j(I),j(I)+1,\dots,i-1\}$. Consider a character $c\in \operatorname{AC}(t_k,t_{k+1})\setminus C$. Let $\alpha$ be the amount of characters from $C$ that are crossed by $c$ between $t_k$ and $t_{k+1}$ in $S$. We show that $c$ crosses at most $\alpha$ characters from $C$ between $t_k$ and $t_{k+1}$ in $S'$, which is enough to show the claim.
        Let $\operatorname{abv}_k$ be the amount of $c'\in C$ such that $c'\prec_{\pi_k}c$. Define $\operatorname{abv}_{k+1}$ by replacing $k$ with $k+1$ in the definition. Equivalently define $\operatorname{abv}_k'$ and $\operatorname{abv}_{k+1}'$ by replacing $\pi_k$ by $\pi_k'$ and $\pi_{k+1}$ by $\pi'_{k+1}$, respectively.
        By the definition of $\assign$ we have that $\operatorname{abv}_k=\operatorname{abv}_k'$ and $\operatorname{abv}_{k+1}=\operatorname{abv}_{k+1}'$. By the pigeonhole principle, $c$ must cross at least $\beta=|\operatorname{abv}_k-\operatorname{abv}_{k+1}|$ characters from $C$ between $t_k$ and $t_{k+1}$ (for $S$ and $S'$). Notice that $\beta$ is the exact amount of crossings for $S'$ and $\beta\le \alpha$ must hold as $\beta$ is a lower bound. Thus $c$ crosses less than or equal characters from $C$ between $t_k$ and $t_{k+1}$ in $S'$ when compared to $S$.
    \end{itemize}
    Putting these inequalities together, we get $\crossings(S')\le \crossings(S)$. It is easy to see that we constructed $S'$ such that it satisfies (2).

    It remains to show (3), i.e. that for each $I'\in\mathcal{I}$ such that $S$ is $I'$-consistent and $\tm(I')\ge \tm(I)$, $S'$ is also $I'$-consistent. We show the claim by contradiction, so assume to the contrary that there exists some $I'$ with $\tm(I')\ge \tm(I)$, $S$ is $I'$-consistent, and $S'$ is not. First, it is worth noting that we only change relative orders of character pairs involving at least one character from $C$. 
    We consider different cases.
    \begin{enumerate}
        \item $\tm(I')=\tm(I)$: Then $\charac(I)\cap \charac(I')=\emptyset$. Hence, we did not change relative orders of character pairs from $\charac(I')$ and we obtain a contradiction.
        \item $\tm(I')>\tm(I)=t_i$. We again consider two cases.
        \begin{enumerate}
            \item $\charac(I')\subseteq \charac(I)$: It follows that $j(I')<i$ and further $j(I')\le j(I)$. As we changed orders for character pairs from $\charac(I')$, we have that $S$ is already not $I'$-consistent, a contradiction.
            \item $\charac(I')\not\subseteq\charac(I)$: Then $j(I')\ge i$. But we only changed orders up to $i$, so $S'$ is already not $I'$-consistent, a contradiction.
        \end{enumerate}
    \end{enumerate}
    Lastly, assume that $S$ is type-2-consistent, we prove that $S'$ is as well. We proceed by contradiction.
    Hence, let there be $(I_1,I_2)$ as in \cref{definition:consistent2} such that $S$ is $(I_1,I_2)$-consistent, but $S'$ is not. It follows that the interval $[t_{j(I)+1},t_i]$ has non-empty intersection with $[\tm(I_1),\tm(I_2)]$. It further follows that $\charac(I_1)\subseteq C$. But if $S$ was already $(I_1,I_2)$-consistent, we did not change relative orders nor positions of characters in $\charac(I_1)$ during the above process, a contradiction.
    
    The proof is concluded by applying the above procedure inductively.
\end{proof}

A similar result with a related proof argument holds for type-2 consistency.
\begin{restatable}{lemma}{typetwolemma}
    \label{lemma:type2}
    Let $(T,\mathcal{C},\mathcal{I}, A)$ be an instance with a solution $S$. We can construct from $S$ a type-2-consistent solution $S'$ such that $\crossings(S')\le \crossings(S)$. If $S$ is type-1-consistent, so is $S'$.
\end{restatable}
\begin{proof}
    Let $S$ be a solution that is not type-2-consistent, i.e.\ there exist interactions $I_1,I_2\in\mathcal{I}$ such that $S$ is not $(I_1,I_2)$-consistent with $i=\tm(I_1),j=\tm(I_2)$ and $C=\charac(I)$. Choose such a pair such that $j-i$ is maximized.
    We construct $S'=(\pi_1',\pi_2',\dots,\pi_\ell')$ such that
    \begin{enumerate}[(1)]
        \item $\crossings(S')\le \crossings(S)$,
        \item $S'$ is $(I_1,I_2)$-consistent, and
        \item for each $I_1',I_2'\in \mathcal{I}$ such that $S$ is $(I_1',I_2')$-consistent, $S'$ is also $(I_1',I_2')$-consistent,
    \end{enumerate}
    as follows: For $k\in \{1,2,\dots,i\}\cup \{j+1,j+2,\dots,\ell\}$ we set $\pi_k'=\pi_k$.
    Now we find a character from $C$ that has the fewest crossings with characters not in $C$ between $t_i$ and $t_j$. Formally, we define $c^*\in C$ such that
    \[c^*=\operatorname*{argmin}_{c\in C}\sum_{k=i}^{j-1}\crossings(\pi_k,\pi_{k+1},\{c\},\operatorname{AC}(t_k,t_{k+1}\setminus C)).\]
    For each $k\in \{i+1,i+2,\dots,j\}$ we do the following. Let $\pi^a,\pi^b$ be such that $\pi_k=\pi^a\star (c^*)\star \pi^b$, where $(c^*)$ is the unit permutation of the set $\{c^*\}$. We set $\pi'_k=\pi^a[\operatorname{AC}(t_k)\setminus C]\star \pi_i[C]\star \pi^b[\operatorname{AC}(t_k)\setminus C]$.
    Informally, we first remove $C$ from $\pi_k$ and then insert the permutation $\pi_i[C]$ at the position of $c^*$.
    We first argue that $\crossings(S')\le \crossings(S)$.
    First, it is clear that $\crossings(\pi'_k,\pi'_{k+1})=\crossings(\pi_k,\pi_{k+1})$ for $k\le i-1$ and $k>j$ as the involved permutations did not change.
    It remains to show that 
    \begin{equation}
        \sum_{k=i}^j\crossings(\pi_k',\pi_{k+1}')\le \sum_{k=i}^j\crossings(\pi_k,\pi_{k+1}).   
    \end{equation}
    We split up the crossings into types, as in the proof of \cref{lemma:type1}; i.e.\ we consider crossings between pairs $C\times C,\operatorname{AC}(t_k,t_{k+1})\times C$, and $\operatorname{AC}(t_k,t_{k+1})\times \operatorname{AC}(t_k,t_{k+1})$, respectively.
    \begin{itemize}
        \item  $\sum_{k=i}^j\crossings(\pi_k',\pi_{k+1}',C)\le \sum_{k=i}^j\crossings(\pi_k,\pi_{k+1},C)$: We have
        \begin{align*}
            \sum_{k=i}^j\crossings(\pi_k,\pi_{k+1},C)&=\sum_{k=i}^j\crossings(\pi_k[C],\pi_{k+1}[C])\\
            &\ge \sum_{k=i}^j\crossings(\pi_k[C\cap \operatorname{AC}(t_{j+1})],\pi_{k+1}[C\cap \operatorname{AC}(t_{j+1})])\\
            &\ge \crossings(\pi_{i}[C\cap \operatorname{AC}(t_{j+1})], \pi_{j+1}[C\cap \operatorname{AC}(t_{j+1})])\\
            &=\crossings(\pi'_{j}[C\cap \operatorname{AC}(t_{j+1})], \pi_{j+1}'[C\cap \operatorname{AC}(t_{j+1})])\\
            &=\sum_{k=i}^j\crossings(\pi_k',\pi_{k+1}',C)
        \end{align*}
        The first inequality holds because $C\cap \operatorname{AC}(t_{i+1})$ is a subset of $C$. The second inequality holds because of \cref{obs:triangle}. The last equality holds because the only crossings between characters from $C$ from $t_{i}$ to $t_{j+1}$ appear between $t_{j}$ and $t_{j+1}$ as the relative order of characters from $C$ in $S'$ is the same from $i$ to $j$.
        \item $\sum_{k=i}^j\crossings(\pi_k',\pi_{k+1}',\operatorname{AC}(t_k,t_{k+1})\setminus C)\le \sum_{k=i}^j\crossings(\pi_k,\pi_{k+1},\operatorname{AC}(t_k,t_{k+1})\setminus C)$: Let $k\in \{i,i+1,\dots, j\}$ and let $C_{\text{ac}}=\operatorname{AC}(t_k,t_{k+1})\setminus C$. We have that
        \begin{align*}
            \crossings(\pi_k',\pi_{k+1}',\operatorname{AC}(t_k,t_{k+1})\setminus C)&=\crossings(\pi_k'[C_{\text{ac}}],\pi_{k+1}'[C_{\text{ac}}])\\
            &=\crossings(\pi_k[C_{\text{ac}}],\pi_{k+1}[C_{\text{ac}}])\\
            &=\crossings(\pi_k,\pi_{k+1},\operatorname{AC}(t_k,t_{k+1})\setminus C).
        \end{align*}
        The key is that $\pi_k[C_{\text{ac}}]=\pi_k'[C_{\text{ac}}]$ and $\pi_{k+1}[C_{\text{ac}}]=\pi_{k+1}'[C_{\text{ac}}]$, as we did not change the relative order of non-interaction characters, that is, characters that are not in $C$.
        As the equality holds for each term, it also holds for the sum.
        \item $\sum_{k=i}^j\crossings(\pi_k',\pi_{k+1}',C, \operatorname{AC}(t_k,t_{k+1})\setminus C)\le \sum_{k=i}^j\crossings(\pi_k,\pi_{k+1},C, \operatorname{AC}(t_k,t_{k+1})\setminus C)$: This inequality holds by definition of $c^*$. Each character from $C$ is involved in the same crossings between $t_i$ and $t_j$.
        The crossings between $t_j$ and $t_{j+1}$ are exactly the same, as relative orders of character pairs in $C\times (\operatorname{AC}(t_{k},t_{k+1})\setminus C)$ are the same in $S$ and $S'$ for $k=j,j+1$.
    \end{itemize}
\end{proof}%
The following is a direct consequence.
\begin{corollary}\label{corollary:crossingminimaltype12}
    For each storyline instance $(T,\mathcal{C},\mathcal{I}, A)$ there exists a crossing-minimum solution $S$ that is type-1-consistent and type-2-consistent.
\end{corollary}
\cref{theorem:consistent3} is the main ingredient for a new ILP formulation given in \cref{section:newformulation}.
It shows that we can in specific cases assume the order of characters $C_a$ above and $C_b$ below an interaction at time step $t_i$ to be equal to the relative order at $t_{i-1}$. This is similar to type-1-consistency, where the relative order of characters in an interaction sometimes can be kept.%
\begin{restatable}{theorem}{propthm}
    \label{theorem:consistent3}
    Let $(T,\mathcal{C},\mathcal{I}, A)$ be a storyline instance. There exists a crossing-minimum solution $S=(\pi_1,\pi_2,\dots,\pi_\ell)$ with the following property.
    For all $t_i\in \{t_2,t_3,\dots,t_\ell\}$ with $|\mathcal{I}(t_i)|=1$, where $\mathcal{I}(t_i)=\{I\}$, the following holds.
    \begin{enumerate}[(1)]
        \item $\exists C_a,C_b:\pi_i=\pi_i[C_a]\star \pi_i[\charac(I)]\star \pi_i[C_b]$,
        \item if $C_a\subseteq \operatorname{AC}(t_{i-1},t_i)$, then $\pi_i[C_a]=\pi_{i-1}[C_a]$,
        \item if $\charac(I)\subseteq \operatorname{AC}(t_{i-1},t_i)$, then $\pi_i[\charac(I)]=\pi_{i-1}[\charac(I)]$, and
        \item if $C_b\subseteq \operatorname{AC}(t_{i-1},t_i)$, then $\pi_i[C_b]=\pi_{i-1}[C_b]$.
    \end{enumerate}
\end{restatable}
\begin{proof}
    Consider a crossing-minimum solution $S^*=(\pi_1^*,\pi_2^*,\dots,\pi_\ell^*)$. We construct a new storyline instance $(T,\mathcal{C},\mathcal{I}', A)$ where $\mathcal{I}'=\mathcal{I}\cup \mathcal{I}_{S^*}$, with $\mathcal{I}_{S^*}$ containing the following interactions.
    For each $t_i\in \{t_2,t_3,\dots,t_\ell\}$ with $|\mathcal{I}(t_i)|=1$ and $\mathcal{I}=\{I\}$, let $C_a,C_b\in \operatorname{AC}(t_i)$ such that $\pi_i^*=\pi_i[C_a]*\pi^*_i[\charac(I)]\star \pi^*_i[C_b]$.
    If $C_a\subseteq \operatorname{AC}(t_{i-1},t_i)$, then add to $\mathcal{I}_{S^*}$ the interaction $I_a$ with $\charac(I_a)=C_a$ and $\tm(I_a)=t_i$.
    If $C_b\subseteq \operatorname{AC}(t_{i-1},t_i)$, then add to $\mathcal{I}_{S^*}$ the interaction $I_b$ with $\charac(I_b)=C_b$ and $\tm(I_b)=t_i$.
    Note that $S^*$ is still a crossing-minimum solution for the new instance. Hence, we apply \cref{lemma:type1} to $S^*$ for the new storyline instance. We obtain a new solution $S$ that is crossing-minimum and type-1-consistent. Type-1-consistency for new interactions implies (2) and (4), type-1 consistence of the original interactions implies (3). Lastly, $S$ is also a crossing-minimum solution of the original instance as $\crossings(S)\le \crossings(S^*)$ by \cref{lemma:type1}, and the statement follows.
\end{proof}

\section{Refining the ILP models}\label{section:newformulation}
We apply our structural insights from \cref{section:structural} to the models (besides the \eqref{cut-formulation}-model) to obtain a new ILP formulation, including a reduction of the number of (LOP) constraints in \cref{ssec:plo} via \cref{theorem:consistent3} and the inclusion of additional symmetry breaking constraints in \cref{section:removesymmetries} via \cref{corollary:crossingminimaltype12}.

\subsection{The Propagated Linear Ordering Model (PLO)}\label{ssec:plo}

For our new formulation, we take the linearized model~\eqref{LIN-ILP} as basis, but remove some of the \eqref{LOP-Con}-constraints for time step $t_i$ as we can make use of propagating the ordering at $t_{i-1}$ by \cref{theorem:consistent3} as follows. 
If $\mathcal{I}(t_i)$ for $i>1$ contains only one interaction $I$, and no characters outside the interaction start at $t_i$ (i.e., $\operatorname{AC}(t_i)\setminus \operatorname{AC}(t_{i-1})\subseteq \operatorname{CI}(t_i)$), we only include a part of the \eqref{LOP-Con}-constraints for time step $t_i$ using a representative character $c_w\in\charac(I)$: %

From the set of \eqref{LOP-Con}-constraints containing at least one character in $\operatorname{AC}(t_i)\setminus \charac(I)$, we keep only those that contain exactly two characters in $\operatorname{AC}(t_i)\setminus \charac(I)$ and the representative character $c_w\in\charac(I)$.
This is sufficient, because we can define the order of the active characters in $t_i$ relative to the order of the characters in the interaction $I$ based on \cref{theorem:consistent3}. Hence, let $c_w$ be a representative character from the set $\charac(I)$, and consider a pair of characters $c_u,c_v\in \operatorname{AC}(t_i)\setminus \charac(I)$. 
By \cref{theorem:consistent3}, if both $c_u$ and $c_v$ are above or below $c_w$, then their relative order can be fixed by their relative order at $t_{i-1}$. Otherwise, their relative order is already given by their relative order to $c_w$. That is, if, e.g., $c_u$ is above $c_w$ and $c_v$ is below $c_w$, then we know that $c_u$ is above $c_v$. To ensure the above, we add the following constraints in addition to the %
\eqref{LOP-Con}-constraints for $c_u$, $c_v$, and $c_w$ at time step $t_i$.
    \begin{align}
        x_{i,u,v}&\ge x_{i-1,u,v}+x_{i,u,w}+x_{i,v,w}-2\tag{PROP-R1} \label{PROP-con1}\\
        x_{i,u,v}&\ge x_{i-1,u,v}+x_{i,w,u}+x_{i,w,v}-2\tag{PROP-R2} \label{PROP-con2} %
\end{align}%
    The two constraints ensure that $c_u$ is above $c_v$ if the requirements are met. By switching $c_u$ and $c_v$, these constraints also ensure the case that $c_u$ is below $c_v$. 

If additionally $\charac(I)\subseteq \operatorname{AC}(t_{i-1})$ we can apply \cref{theorem:consistent3} (3) to further reduce the number of those \eqref{LOP-Con}-constraints, whose triples are taken from the set $\charac(I)$:
In this case, we do not add any of the \eqref{LOP-Con}-constraints for the characters in $I$, but instead for each pair $c_u,c_v\in \charac(I)$, we add the following constraint ensuring that the relative order of $c_u$ and $c_v$ is the same for $t_i$ and $t_{i-1}$.
    \begin{align}
        x_{i,u,v}=x_{i-1,u,v}\tag{PROP-I} \label{PROP-conI}
    \end{align}
 If both reductions for \eqref{LOP-Con} apply, we get a quadratic rather than cubic number of constraints for $t_i$.
We call this formulation \emph{propagated linear order} (PLO). Note that this idea of reducing the number of \eqref{LOP-Con}-constraints also works for any of the other standard ILP models. 

\begin{restatable}{theorem}{plotransitivity}\label{thm:plotransitivity}
    Every optimal solution to the formulation (PLO) corresponds to a crossing minimum storyline drawing. %
\end{restatable}
\begin{proof}
    We only have to show that solutions to PLO satisfy transitivity constraints for the ordering variables, as we include the constraint set \eqref{CR-Con} and \eqref{LIN-ILP}. Optimality follows from \cref{theorem:consistent3}.
    We show this by induction on $i$ -- the time steps. For the base case of $i=1$ the variables $x_{1,u,v}$ for $c_u,c_v\in \operatorname{AC}(t_1)$ satisfy transitivity, as we include all \eqref{LOP-Con} constraints for $t_1$.
    For the induction step, assume that transitivity is satisfied for all time steps $t_{i'}$ with $i'<i$. We show that transitivity is satisfied for $t_i$. We perform a case distinction.
    \begin{itemize}
        \item If $|\mathcal{I}(t_i)|>1$ or $(\operatorname{AC}(t_i)\setminus \operatorname{AC}(t_{i-1}))\not\subseteq \operatorname{CI}(t_i)$, transitivity is clearly satisfied, as we again include all \eqref{LOP-Con}-constraints for this time step. 
        \item Otherwise, there is exactly one interaction $I$ at time $t_i$. Consider now three distinct characters $c_u,c_v,c_w\in \operatorname{AC}(t_i)$. We show transitivity between this triple for $t_i$ by considering memberships with respect to $I$.
        \begin{itemize}
            \item If $\{c_u,c_v,c_w\}$ contains at least one character from $\charac(I)$ but $\{c_u,c_u,c_w\}\not\subseteq \charac(I)$ then transitivity for these three characters is achieved because of the \eqref{TREE-Con}-constraints.
            \item If $\{c_u,c_v,c_w\}\cap \charac(I)=\emptyset$, then we know that $\{c_u,c_v,c_w\}\subseteq \operatorname{AC}(t_{i-1})$. Thus, transitivity between this triple follows from \eqref{PROP-con1} and \eqref{PROP-con2} and because, by the induction hypothesis, the same triple satisfies transitivity for time step $t_{i-1}$.
            \item If $\{c_u,c_v,c_w\}\subseteq \charac(I)$, we again have two cases. If $\charac(I)\not\subseteq \operatorname{AC}(t_{i-1})$, transitivity is clearly satisfied as we add all \eqref{LOP-Con}-constraints between triples in $\charac(I)$ for this time step. Otherwise, transitivity for this triple follows from the induction hypothesis and \eqref{PROP-conI}.\qedhere
        \end{itemize}
    \end{itemize}
\end{proof}

\subsection{Symmetry Breaking Constraints}\label{section:removesymmetries}
We introduce the set (SBC) of \emph{symmetry breaking constraints} that are based on \cref{corollary:crossingminimaltype12} and might improve the solving process of the models, as they constitute equalities:
\begin{itemize}
    \item We can assume that a crossing-minimum solution is type-1 consistent. Thus let $I\in\mathcal{I}$ with $t_i=\tm(I)$ and let $j(I)$ be defined as in \cref{definition:consistent1}.  For all pairs $c_u,c_v\in \charac(I)$ and all $j(I)\le k<i$ we can add the following constraint enforcing type-1 consistency:
    \begin{align}
        x_{k,u,v}=x_{i,u,v}\tag{SBC-1}\label{symcons1}
    \end{align}
    \item We can assume that a crossing-minimum solution is type-2 consistent. Thus, let $I_1,I_2\in\mathcal{I}$ be two distinct interactions satisfying the properties of \cref{definition:consistent2}. Let $i=\tm(I_1)$ and $j=\tm(I_2)$.
    For all $i<k<j$, all pairs $c_u,c_v\in \charac(I_1)$, and all $c_w\in \operatorname{AC}(t_k)\setminus \charac(I_1)$ %
    we add the following constraint, enforcing type-2 consistency:
    \begin{align}
        x_{k,u,w}=x_{k,v,w}\tag{SBC-2}\label{symcons2}
    \end{align}
\end{itemize}

\section{Implementation}\label{section:implementation}
In this section, we discuss relevant implementation details and new heuristic-based approaches to improve our algorithms.
\subsection{Initial Heuristic}\label{section:initial}

Our initial heuristic follows a decomposition strategy. The original problem instance is first split into smaller subproblems, each comprising $\ensuremath{\hat{\ell}}$ many consecutive time steps, which are of such a size that it is possible to quickly compute a crossing minimal solution for this subproblem.
The initial heuristic solution $S_h$ is computed by optimally solving these subsequent ``slices'' of the original problem instance and piecing together a global solution. This works as follows.
First, a crossing minimal solution $S_1 = (\pi_1, \dots, \pi_{\hat{\ell}})$ is computed for the first $\ensuremath{\hat{\ell}}$ many time steps of the original problem using the \hyperref[section:newformulation]{PLO} ILP formulation. The first $\ensuremath{\hat{s}}$ many layers of $S_1$ are assigned to the initial heuristic solution $S_h$. In the next iteration, we compute an optimal solution for time steps $\ensuremath{\hat{s}},\dots,\ensuremath{\hat{s}}+\ensuremath{\hat{\ell}}-1$ while enforcing the ordering of layer $\hat{s}$ by fixing the ordering variables accordingly. Again, we fix the first $\ensuremath{\hat{s}}$ many permutations of this solution as the permutations of the corresponding time steps in the heuristic solution $S_h$. This is continued until the heuristic solution is computed for all time steps $t_1,t_2,\dots,t_\ell$.

We choose $\ensuremath{\hat{s}} < \ensuremath{\hat{\ell}}$ in order to include as much information about the global problem into the sliced subproblem as possible (in the form of the time steps that are about to come), while achieving a good tradeoff between time needed to solve all subproblems and the number of crossings of $S_h$. Experimental evaluation showed that $\ensuremath{\hat{s}} = 5, \ensuremath{\hat{\ell}} = 30$ yields a good tradeoff between runtime and solution quality.

\subsection{Rounding and Local Improvement Heuristics}\label{section:roundingimp}
We propose a rounding heuristic that exploits fractional LP-solutions. %
Furthermore, we try to improve these solutions as well as incumbent solutions found by the solver software by proposing three local improvement heuristics {\textcolor{lipicsGray}{\bfseries\sffamily Rem-DC}, {\textcolor{lipicsGray}{\bfseries\sffamily Push-CR}, and {\textcolor{lipicsGray}{\bfseries\sffamily SL-Bary}.

\subparagraph{Rounding Heuristic.} We propose a strategy to round fractional solutions of the ordering variables $x_{i,u,v}$ to valid integer solutions corresponding to a drawing of the storyline instance.  
In particular, we compute permutations $\pi_i$ for $t_i$, going from $i=1$ to $\ell$ in this order, and convert them to integer solutions of ordering variables in the natural way. This works as follows. First, for each interaction $I\in \mathcal{I}(t)$ we compute a permutation $\pi_I$ of its characters $\charac(I)$. For this we compute for each $c\in \charac(I)$ the value $\ensuremath{d^-}(c)$ which is computed as $A+B$ where $A$ is the number of $c'\in\charac(I)$ such that $x_{i,c',c}>0.5+\epsilon$ and $B$ is the number of $c'$ such that $|x_{i,c,c'}-0.5|\le \epsilon$, $\{c,c'\}\subseteq \operatorname{AC}(t_{i-1})$, and $c'\prec_{\pi_{i-1}} c$\footnote{This is useful as setting all ordering variables to $0.5$ is a valid solution to the LP relaxation of most considered models and thus ordering variables often assume this value}. Clearly, $B$ is only positive if $i>1$. Then $\pi_I$ is computed by sorting $\charac(I)$ by their $\ensuremath{d^-}$-values. If the model contains symmetry breaking constraints, this sometimes leads to an infeasible ordering. In this case, we find an order $\pi_I$ that also satisfies the symmetry breaking constraints imposed by $\pi_{i-1}$ (some pairs of characters must have the same relative order in $\pi_i$ and $\pi_{i-1}$) as follows. We construct $\pi_{I}$ iteratively by always selecting the character $c\in \charac(I)$ that has the smallest value $\ensuremath{d^-}(c)$ and, furthermore, there is no character $c'$ which was not selected yet and needs to precede $c$ according to $\pi_{i-1}$ and the symmetry constraints. This is implemented in quasi-linear time using a priority queue.

\subparagraph{Local Improvement Heuristics.}
\begin{description}
    \item[Rem-DC (remove double crossings)] This heuristic finds pairs of characters that cross twice, and both crossings can be removed without increasing the total number of crossings. Formally, this is possible for a drawing $S$ and two characters $c,c'$ if there exist $1\le i<j\le \ell$ with $j-i>1$ such that
\begin{itemize}
    \item $c$ and $c'$ cross between $t_i$ and $t_{i+1}$, and $t_{j-1}$ and $t_{j}$, and
    \item for all $k\in \mathbb{N}$ with $i<k<j$, $c$ and $c'$ either belong to the same interaction in $t_k$, or they both are in no interaction for $t_k$.
\end{itemize}
Then for all $k$ as above we can exchange $c$ and $c'$ in $\pi_k$. This removes the double-crossing between $c$ and $c'$ and further does not introduce new crossings.
    \item[Push-CR] This heuristic proceeds from $i=2,\dots,\ell$ in this order and tries to push crossings between $\pi_{i-1}$ and $\pi_i$ forward by one time step: Let $C$ be a maximal set of characters such that (1) all characters in $C$ appear consecutively in $\pi_i$, (2) $C\subseteq \operatorname{AC}(t_{i-1},t_i)$, and (3) all characters in $C$ either appear in the same interaction in $t_i$ or no character in $C$ is part of an interaction. For each such set of characters $C$ we replace in $\pi_i$, $\pi_i[C]$ by $\pi_{i-1}[C]$. By similar arguments as in \cref{section:structural} this never increases the number of crossings.
    \item[Bary-SL] Lastly, we describe a variant of the barycenter heuristic~\cite{DBLP:journals/tsmc/SugiyamaTT81} for storylines that iteratively improves a storyline drawing by updating $\pi_i$ for $1\le i\le \ell$ based on $\pi_{i-1}$ and $\pi_{i+1}$ or one of them if not both exist. It is only applied to $\pi_i$ if $|\mathcal{I}(t_i)|=1$. Informally, we say that a pair of characters $c,c'$ is comparable if $c$ and $c'$ have the same relative order in $\pi_{i-1}$ and $\pi_{i+1}$. We compute an ordering $\pi_i$ such that most comparable pairs have the same relative order in $\pi_{i-1},\pi_i,\pi_{i+1}$ as follows. We compute the directed auxiliary graph $G_C$ whose vertex set is a subset $S$ of $\operatorname{AC}(t_i)$ and which contains an arc from $c$ to $c'$ for each comparable pair $c$ and $c'$ such that $c$ is before $c'$ in $\pi_{i-1}$ and $\pi_{i+1}$. Then, an order of $S$ is built by iteratively selecting the vertex from $G_C$ with the fewest incoming arcs. We also ensure that characters $c$ that are not part of $I$ are above or below the characters in $I$ depending on which option leads to fewer crossings between $c$ and $\charac(I)$ with respect to the considered time steps $t_{i-1},t_i,t_{i+1}$. The algorithm computing the order based on the graph $G_C$ is then applied to the characters in the interaction yielding $\pi_I$, and those not in the interaction yielding $\pi_C$, respectively. The ordering $\pi_C$ is inserted into the maximum position of $\pi_I$ such that all characters before $\pi_C$ ``prefer'' being above the interaction with regard to crossings with $\charac(I)$.
The new $\pi_i$ is only accepted if it decreases the number of crossings.
\end{description}

Both {\textcolor{lipicsGray}{\bfseries\sffamily Bary-SL}} and {\textcolor{lipicsGray}{\bfseries\sffamily Push-CR}} are applied successively to layers $2,\dots,\ell$. This is repeated five times and applied to valid integer solutions found by the solver and the rounding heuristic described above. If enabled, the rounding heuristic is applied to every LP solution found by the solver. {\textcolor{lipicsGray}{\bfseries\sffamily Rem-DC}} is applied five times to each pair of characters.

\subsection{Max-Cut Implementation Details}\label{section:maxcutimplementation}
Since the original implementation of Gronemann et al.~\cite{GronemannJLM16} is not available, we provide our own implementation that was optimized beyond their algorithm. %
After reading the input, we first find an initial starting solution by applying adapted barycenter techniques as described in~\cite{GronemannJLM16}. %
We start the root relaxation with the objective function and the tree constraints (\ref{TREE-MC}) as the only constraints, %
and start separating the odd cycle (\ref{CYC-MC}) (as suggested by Charfreitag et al.~\cite{CJMM2022}) and the (\ref{LOPC-Con}) constraints. 
The (\ref{LOPC-Con}) constraints are separated by complete enumeration.
Whenever a new LP solution is available, all nonbinding inequalities are eliminated, and we try to exploit the information in the (fractional) solution in order to obtain a better %
incumbent solution. 

The root phase ends when no violated inequalities are found. Then the branch-and-cut phase is started by changing the variable types from continuous in the interval $[0,1]$ to binary. %
In the Gurobi ``MIPSOL'' callbacks at branch-and-cut nodes with an integer solution, we check if the integral solution is the characteristic vector of a storyline drawing. If so and if the number of crossings is lower than the one of the incumbent solution, the latter is updated, otherwise, the exact (\ref{CYC-MC}) and (\ref{LOPC-Con}) separators are called to provide violated inequalities that are passed to Gurobi as \emph{lazy constraints}.
In the Gurobi ``MIPNODE'' callbacks it is tried to exploit the fractional solution for a possible update of the incumbent solution.

\subsection{Implementation of the ILP Models}\label{section:ilpimplementation}
The models \eqref{QDR-ILP:c}, \eqref{LIN-ILP}, and (PLO) include many  symmetries regarding ordering variables and crossing variables. 
For each $t_i\in T$ and pair of characters $c_u,c_v\in \operatorname{AC}(t_i)$ we only keep the ordering variables $x_{i,u,v}$ and crossing variables $y_{i,u,v}$ with $u<v$. The constraints are adjusted with standard projections~\cite{GroetschelJR1985,GronemannJLM16}. 

We take the linearized model \eqref{LIN-ILP} as basis for our refined ILP model described in Section~\ref{section:newformulation}, because preliminary experiments showed no performance gain from refining \eqref{QDR-ILP:c} instead. Further, the linearized model \eqref{LIN-ILP} is competitive with the max-cut approach when implemented in Gurobi. Therefore, we decided on refining the linearized model that is simpler to implement and more accessible when compared with the max-cut approach.  %
Furthermore, implementing any of the ILP models naively includes up to $\mathcal{O}(\ell n^3)$ \eqref{LOP-Con}-constraints in the model. We have experimented with adding these constraints during a cutting-plane approach and also by including them into the Gurobi solver as \emph{lazy constraints}, i.e., constraints that the solver can decide to include at later stages during the solving process. 
We decided to always add \eqref{LOP-Con} as lazy constraints, as this leads to the best performance.
Hence, we consider the following algorithms for our experimental evaluation.
\begin{itemize}
    \item\mc: the max-cut formulation (as a baseline) implemented as described in \cref{section:maxcutimplementation}
    \item\lin: the linearized model \eqref{LIN-ILP} with \eqref{LOP-Con}-constraints included as lazy constraints
    \item\qdr: the quadratic model \eqref{QDR-ILP:c} with \eqref{LOP-Con}-constraints included as lazy constraints
    \item\plo: the PLO formulation with \eqref{LOP-Con}-constraints included as lazy constraints
\end{itemize}
The latter three algorithms are by default extended with the symmetry breaking constraints described in \cref{section:removesymmetries} (\sbc), the initial heuristic from \cref{section:initial} (\init), and the rounding and local improvement heuristics from \cref{section:roundingimp} (\rnd). This is not done for \mc, as it should serve as a state-of-the-art baseline and allow comparison with Gronemann et al.~\cite{GronemannJLM16}.

\section{Experiments and Evaluation}\label{sec:experiments}

In our experimental evaluation, we are interested in the following research questions.
\begin{description}
    \item[Q1:] Does the algorithm \plo\ based on our new ILP model dominate the state-of-the-art model \mc ? Will we be able to solve hard instances that have not been solved to optimality before? How do the various algorithms compare to each other?
    \item[Q2:] What effect do the structural insights have when applied to the \lin-formulation? 
    \item[Q3:] What is the effect of the newly introduced components \sbc, \init, and \rnd?
\end{description}

In the following we describe our experimental setup, our benchmark instances, and the results of our study. We also provide our results and analysis on  \href{https://doi.org/10.17605/OSF.IO/3BUA2}{\texttt{https://osf.io/3bua2/}}.

\subsection{Setup}
Systems employed for all experiments have AMD EPYC 7402, 2.80GHz 24-core CPUs and 1024GB of RAM, running Ubuntu 18.04.6 LTS; experiments were run using a single thread. 

\mc\ is implemented in \texttt{C} and compiled with \texttt{gcc 7.5.0}, \texttt{GNU make 4.1}, and flag \texttt{-O3}, all remaining code is written in \texttt{C++17}, compiled with \texttt{cmake 3.10.2} and \texttt{g++ 11.4.0} in \texttt{Release} mode. To solve the ILPs we used Gurobi 11.0.1. The time limit is 3600s (same as Gronemann et al.~\cite{GronemannJLM16}) and the memory limit is 16GB for all experiments. We do not know the memory limit for Gronemann, however memory was certainly not our limiting factor.
The time for the initial heuristic is negligible ($<1$\% of the overall runtime), so it is not counted towards the solving time.

To mitigate performance variability, we ran each instance-setting combination with five different seeds provided to Gurobi; data displayed below corresponds to the seed with the median runtime. With this, an instance counts as ``solved in the time limit'', if the majority of the five runs does not time out. 
Source code for the new formulations is available on  \href{https://doi.org/10.17605/OSF.IO/3BUA2}{\texttt{https://osf.io/3bua2/}}.

\subsection{Test Data}
The instances used in our computational study are taken from the literature~\cite{di2020storyline, GronemannJLM16,MovieDataset2013,DoblerNSVW23}. However, we also present a new data set, including existing instances, in a specifically designed storyline data format, together with tools for transformation and visualization of the storyline layouts on  \href{https://doi.org/10.17605/OSF.IO/3BUA2}{\texttt{https://osf.io/3bua2/}}. We also provide data on best known crossing numbers. %

The existing instances from Gronemann et al.~\cite{GronemannJLM16} consist of three book instances from the Stanford GraphBase database~\cite{knuth1993stanford}, i.e.,\ \textit{Anna Karenina} (anna), \textit{Les Misérables} (jean) and \textit{Adventures of Huckleberry Finn} (huck), and the movie instances TheMatrix, Inception, and StarWars. The instances gdea10, gdea20 from Dobler et al.~\cite{DoblerNSVW23} consist of publication data from 10 (resp.~20) authors from the GD conference. The publication instances ubiq1, ubiq2 are from Di Giacomo et al.~\cite{di2020storyline}. Furthermore, anna and jean are split up into slices of 1-4 consecutive chapters as was done by Gronemann et al.~\cite{GronemannJLM16}\footnote{We could not replicate this process fully equivalently, as sometimes our optimal crossing numbers are different to those of Gronemann et al.~\cite{GronemannJLM16}.}.
The new instances consist of scenes from nine blockbuster movies, namely Avatar, Back to The Future, Barbie, Forrest Gump, Harry Potter 1, Jurassic Park, Oceans 11, Oppenheimer, and Titanic. 

In all 59 resulting instances, characters are active from their first interaction to their last interaction, and most instances have one interaction per time step.

\subsection{Evaluation}

Several instances could be solved within 30s by all algorithms, others could not be solved within the time limit by any of the algorithms. 
The three instances TheMatrix, Inception, and StarWars used commonly in heuristic storyline visualization were all solved within 450ms.
We exclude all these instances and focus on the 23 challenging instances that remain. Out of these, the maximum number of characters is 88, and the maximum number of layers is 234. A table showing detailed statistics of test data and executions of each algorithm is given in \cref{apx:table}.

\subparagraph{Answering \qone\ and \qtwo.}
\begin{figure}
    \centering
    \includegraphics{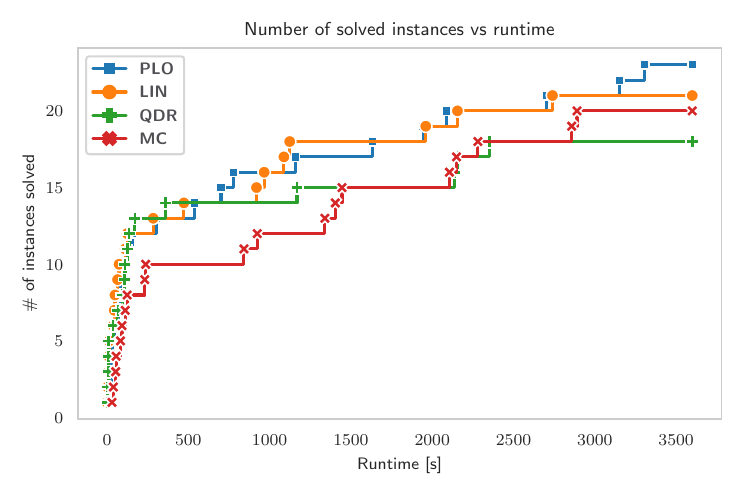}
    \caption{Number of instances per time limit given.} 
    \label{fig:survival_allon}
\end{figure}
\begin{figure}
    \centering
    \includegraphics{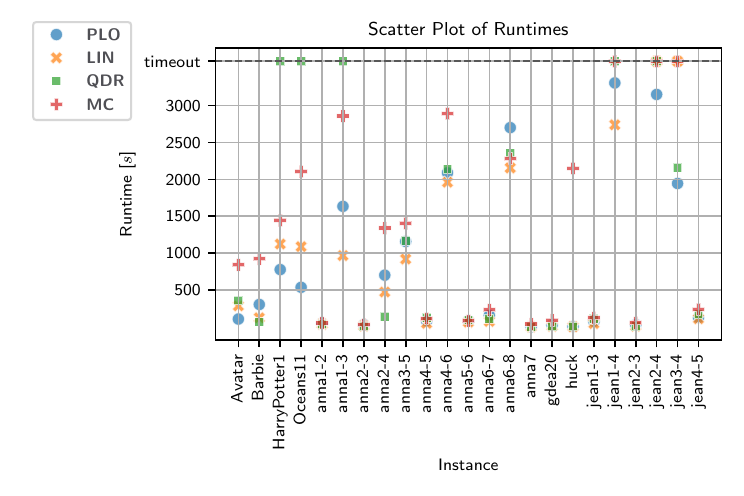}
    \caption{Runtimes of algorithms per instance.} 
    \label{fig:scatterplot_allon}
\end{figure}
Figure~\ref{fig:survival_allon} displays the number of instances solved over time for each algorithm. We observe that the algorithms differ in their ability to solve challenging instances: \plo\ solves the most, followed by \lin\ and \mc, with \qdr\ last. In fact, \plo\ solves one instance that cannot be solved by any other algorithm, and additionally solves six instances that could not be solved by Gronemann et al.~\cite{GronemannJLM16} and three more than \mc\ within the same time limit. Hence, we answer the first part in \qone\ positively. For further illustration, the exact runtimes per instance are also shown in \cref{fig:scatterplot_allon}.

Answering \qtwo, the structural insights as applied in \plo\ reduce the number of constraints by a factor of five on average, comparing \lin\ and \plo. More so, they enhance Gurobi's capabilities of strengthening the LP relaxation, as the two instances not solved by \lin\ are solved by \plo\ in the root, while \lin\ starts branching early and times out.
\qdr\ enters branching in all 23 instances, \plo\ in two, \mc\ in three, \lin\ in seven instances. \plo\ solves 21 out of 23 instances in the root, the remaining two with branching.

Furthermore, we computed the \emph{speedup factor} of \plo, \lin, and \qdr, when compared with \mc\ on instances where both respective algorithms did not time out. 
This factor is the runtime of \mc\ divided by the runtime of, e.g., \plo. The geometric means of these values are 2.6 for \plo, 3.2 for \lin, and 2.7 for \qdr. Hence, our new algorithms are 2.6--3.2 times faster than the state-of-the-art algorithm \mc.%

\subparagraph{Ablation study to answer \qthree.}

\begin{figure}[tbp]
    \centering
    \includegraphics{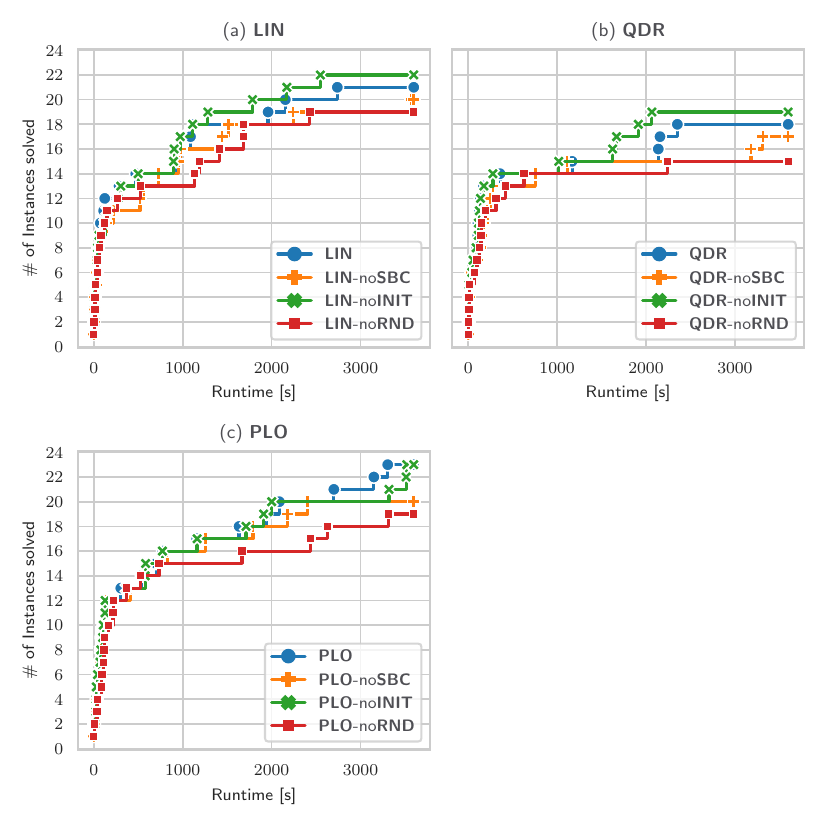}
    \caption{%
    Number of instances solved per time limit given, broken down by algorithm. Algorithms are compared with their counterparts where exactly one component is disabled.}
    \label{fig:survival_ablation}
\end{figure}
\begin{figure}
     \centering
     \includegraphics[height=0.92\textheight]{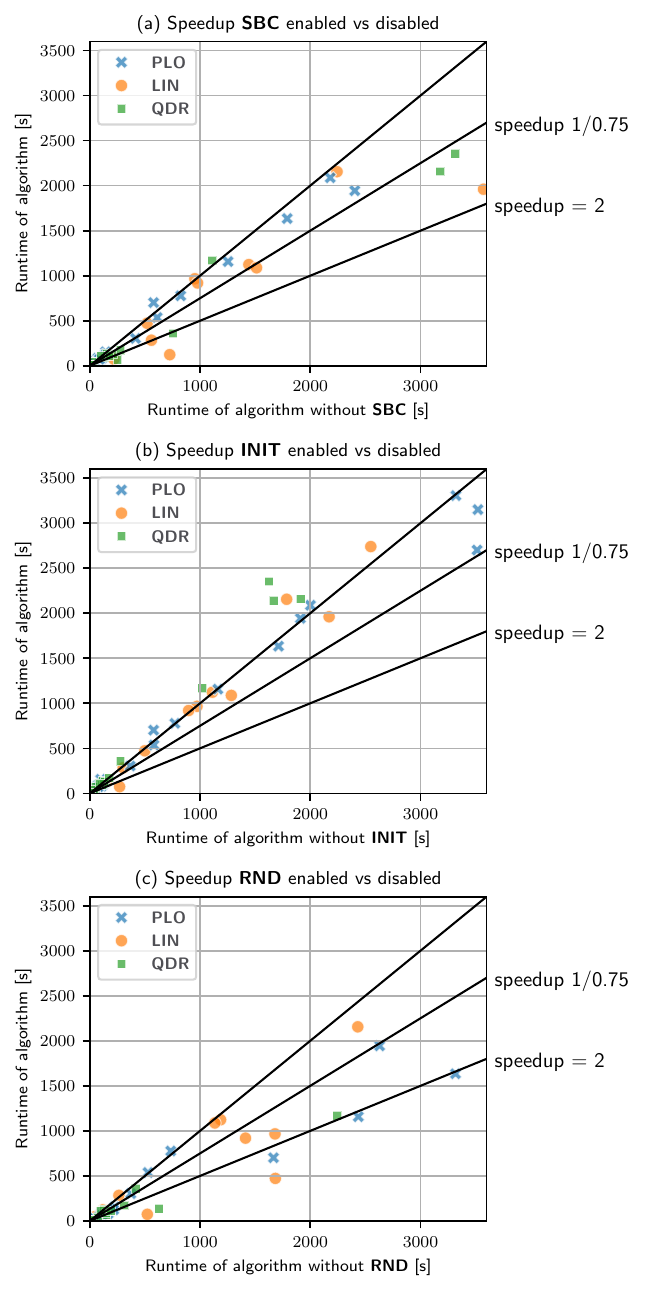}
     \caption{Comparison of runtime improvement for each instance when disabling the respective component. Dots below the top line ($y=x$) depict a runtime improvement. Dots below the remaining lines depict a speedup factor by at least $1/0.75$ and $2$, respectively.}
    \label{fig:speedup_graphs}
\end{figure}
\begin{table}[ht]
\centering
\caption{Geometric means of the speedup factor of each baseline algorithm vs.\ its counterpart where the respective component is disabled.} %
\label{table:speedup}
\renewcommand{\arraystretch}{1.2}
\begin{tabular}{lrrr} %
\toprule
\multicolumn{1}{c}{speedup factor} & \multicolumn{1}{l}{\plo} & \multicolumn{1}{l}{\lin} & \multicolumn{1}{l}{\qdr} \\ \midrule
 \sbc & 1.12 & 1.35 & 1.42  \\ %
 \init & 1.05 & 1.09 & 0.97  \\ %
 \rnd & 1.50 & 1.37 & 1.52 \\ %
 \bottomrule
\end{tabular}
\end{table}

We conduct an ablation study to discern the impact each of the methods proposed in \cref{section:newformulation,section:implementation} has on the algorithms' performance. To this end, we enable all the proposed methods as the baseline configuration for \plo, \lin, and \qdr, namely \sbc, \init~ and \rnd. Then, each component is disabled one at a time to measure the component's impact on overall performance. 
In \cref{table:speedup} we present the speedup factors of the algorithms vs.\ their counterparts with the specific component disabled (see also \cref{fig:speedup_graphs} for more details).
From this table, we conclude that \sbc\ and the \rnd\ are beneficial for all algorithms, while \init\ has a small to no noticeable impact.
This is further supported by \cref{fig:survival_ablation}, which shows that disabling \sbc\ or \rnd, results in all the formulations solving fewer instances (curves with {\sffamily no\sbc} and {\sffamily no\rnd} are below the baseline). Further, \cref{fig:speedup_graphs} depicts the speedup factors per instance.
This is because \sbc\ introduces equalities between two variables, and hence improve presolving capabilities and reduce the search space that solvers have to explore. The heuristics of \rnd\ help the solver find optimal solutions early in the  process. This answers \qthree.

\subparagraph{Further observation.}
By inspecting \cref{table:runtimes} in more detail, we can make the following interesting observations.
\begin{itemize}
    \item \qdr\ enters branching on every instance. This results in \qdr\ solving very small instances faster than the other algorithms, however it struggles to solve larger instances.
    \item The larger the instance, the fewer constraints \plo\ has compared to \lin.
    \item Most of the instances can be solved within the root of the branch-and-bound tree for most formulations.
\end{itemize}

\begin{figure}[tb]
    \centering
    \begin{subfigure}{\textwidth}
        \includegraphics[width=\textwidth]{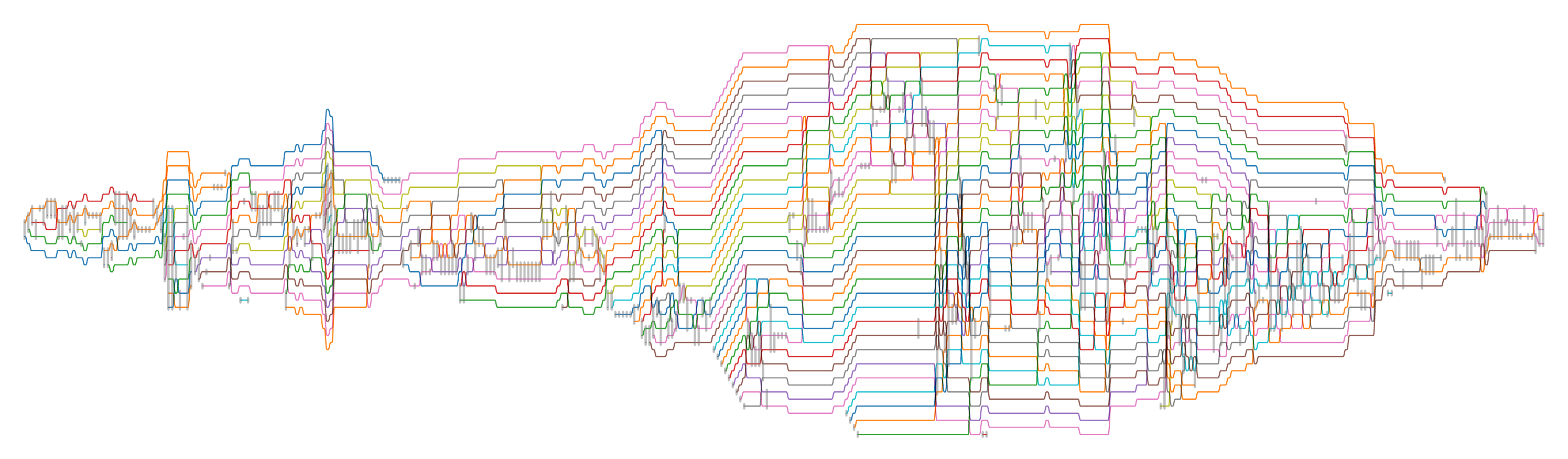}
        \caption{Solution with 765 crossings computed in \qty{0.57}{\second} by a greedy heuristic.}
    \end{subfigure}
    \begin{subfigure}{\textwidth}
        \includegraphics[width=\textwidth]{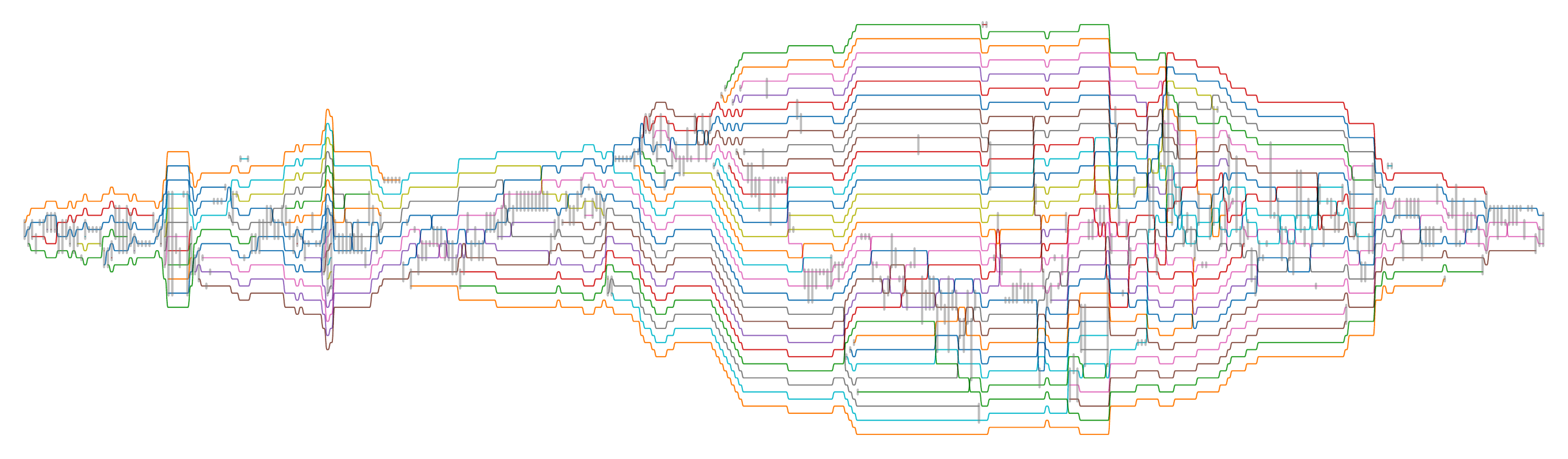}
        \caption{With the minimum number 244 of crossings computed in $\approx7$ hours.} %
    \end{subfigure}    
    \caption{Storyline of Les Mis\'erables (\texttt{jean}) with 80 characters and 402 layers.}
    \label{fig:jean}
\end{figure}

\subparagraph{Large Instances.} Finally, we demonstrate that our implementations are capable of solving even very large instances to proven optimality: %
Figure~\ref{fig:jean}(a) shows the raw drawing of the data for Victor Hugo’s Les Mis\'erables~\cite{lesmiserables} as provided in the data file \texttt{jean.dat} of the Stanford GraphBase~\cite{knuth1993stanford}. After roughly 7 hours of single thread computation, we obtained the proven crossing minimum layout shown in Figure~\ref{fig:jean}(b).

\section{Conclusion and Future Work}
As shown in our experimental study, our new methods and algorithms dominate the state-of-the-art algorithms and are able to solve large instances to optimality, while the newly introduced improvements are beneficial towards all considered formulations. We observe two directions for future work.
\begin{itemize}
    \item Our new components for improvement could be implemented into the max-cut formulation. However, initial experiments have shown that the simple linearized formulation \eqref{LIN-ILP} performs comparably to the more complex max-cut formulation, hence we expect that this will result in a negligible or no improvement over our proposed formulations.
    \item Out of our 59 instances (see  \href{https://doi.org/10.17605/OSF.IO/3BUA2}{\texttt{https://osf.io/3bua2/}}) we were able to solve 55 when increasing the time limit. The remaining four  unsolved instances should pose a challenge to engineer new exact methods for crossing minimization in storylines.
\end{itemize}

\bibliography{literature}

\begin{thebibliography}{10}

\bibitem{axpab-hiudh-22}
Vanessa~Pe{\~{n}}a Araya, Tong Xue, Emmanuel Pietriga, Laurent Amsaleg, and Anastasia Bezerianos.
\newblock Hyperstorylines: Interactively untangling dynamic hypergraphs.
\newblock {\em Inf. Vis.}, 21(1):38--62, 2022.
\newblock \href {https://doi.org/10.1177/14738716211045007} {\path{doi:10.1177/14738716211045007}}.

\bibitem{bett-gd-99}
Giuseppe~Di Battista, Peter Eades, Roberto Tamassia, and Ioannis~G. Tollis.
\newblock {\em Graph Drawing: Algorithms for the Visualization of Graphs}.
\newblock Prentice-Hall, 1999.

\bibitem{DBLP:journals/informs/BuchheimWZ10}
Christoph Buchheim, Angelika Wiegele, and Lanbo Zheng.
\newblock Exact algorithms for the quadratic linear ordering problem.
\newblock {\em {INFORMS} J. Comput.}, 22(1):168--177, 2010.
\newblock \href {https://doi.org/10.1287/IJOC.1090.0318} {\path{doi:10.1287/IJOC.1090.0318}}.

\bibitem{CJMM2022}
Jonas Charfreitag, Michael J{\"u}nger, Sven Mallach, and Petra Mutzel.
\newblock {\em McSparse: Exact Solutions of Sparse Maximum Cut and Sparse Unconstrained Binary Quadratic Optimization Problems}, pages 54--66.
\newblock Proc. Symposium on Algorithm Engineering and Experiments (ALENEX'22), 2022.
\newblock \href {https://doi.org/10.1137/1.9781611977042.5} {\path{doi:10.1137/1.9781611977042.5}}.

\bibitem{DoblerNSVW23}
Alexander Dobler, Martin N{\"{o}}llenburg, Daniel Stojanovic, Ana{\"{\i}}s Villedieu, and Jules Wulms.
\newblock Crossing minimization in time interval storylines.
\newblock {\em CoRR}, abs/2302.14213, 2023.
\newblock \href {https://doi.org/10.48550/ARXIV.2302.14213} {\path{doi:10.48550/ARXIV.2302.14213}}.

\bibitem{f-fshctcr-04}
Michael Forster.
\newblock A fast and simple heuristic for constrained two-level crossing reduction.
\newblock In J{\'{a}}nos Pach, editor, {\em Proc. Graph Drawing and Network Visualization (GD'04)}, volume 3383 of {\em LNCS}, pages 206--216. Springer, 2004.
\newblock \href {https://doi.org/10.1007/978-3-540-31843-9_22} {\path{doi:10.1007/978-3-540-31843-9_22}}.

\bibitem{gj-cnn-83}
Michael~R. Garey and David~S. Johnson.
\newblock Crossing number is {NP}-complete.
\newblock {\em SIAM J. on Algebraic and Discrete Methods}, 4(3):312--316, 1983.
\newblock \href {https://doi.org/10.1137/0604033} {\path{doi:10.1137/0604033}}.

\bibitem{di2020storyline}
Emilio~Di Giacomo, Walter Didimo, Giuseppe Liotta, Fabrizio Montecchiani, and Alessandra Tappini.
\newblock Storyline visualizations with ubiquitous actors.
\newblock In David Auber and Pavel Valtr, editors, {\em Proc. Graph Drawing and Network Visualization (GD'20)}, volume 12590 of {\em LNCS}, pages 324--332. Springer, 2020.
\newblock \href {https://doi.org/10.1007/978-3-030-68766-3_25} {\path{doi:10.1007/978-3-030-68766-3_25}}.

\bibitem{GronemannJLM16}
Martin Gronemann, Michael J{\"{u}}nger, Frauke Liers, and Francesco Mambelli.
\newblock Crossing minimization in storyline visualization.
\newblock In Yifan Hu and Martin N{\"{o}}llenburg, editors, {\em Proc. Graph Drawing and Network Visualization (GD'16)}, volume 9801 of {\em LNCS}, pages 367--381. Springer, 2016.
\newblock \href {https://doi.org/10.1007/978-3-319-50106-2_29} {\path{doi:10.1007/978-3-319-50106-2_29}}.

\bibitem{GroetschelJR1985}
Martin Gr{\"{o}}tschel, Michael J{\"{u}}nger, and Gerhard Reinelt.
\newblock Facets of the linear ordering polytope.
\newblock {\em Math. Program.}, 33(1):43--60, 1985.
\newblock \href {https://doi.org/10.1007/BF01582010} {\path{doi:10.1007/BF01582010}}.

\bibitem{DBLP:reference/crc/HealyN13}
Patrick Healy and Nikola~S. Nikolov.
\newblock Hierarchical drawing algorithms.
\newblock In Roberto Tamassia, editor, {\em Handbook on Graph Drawing and Visualization}, chapter~13, pages 409--453. Chapman and Hall/CRC, 2013.

\bibitem{h-sfwk-22}
Tim Herrmann.
\newblock {Storyline-Visualisierungen für wissenschaftliche Kollaborationsgraphen}.
\newblock Master's thesis, Universität Würzburg, 2022.
\newblock URL: \url{https://www1.pub.informatik.uni-wuerzburg.de/pub/theses/2022-herrmann-masterarbeit.pdf}.

\bibitem{lesmiserables}
Victor Hugo.
\newblock {\em Les {M}is{\'e}rables}.
\newblock Jules Rouff et Cie editeurs, Paris, 1862.

\bibitem{knuth1993stanford}
Donald~E. Knuth.
\newblock {\em The Stanford GraphBase - a platform for combinatorial computing}.
\newblock {ACM}, 1993.

\bibitem{knpss-mcsv-15}
Irina Kostitsyna, Martin N{\"{o}}llenburg, Valentin Polishchuk, Andr{\'{e}} Schulz, and Darren Strash.
\newblock On minimizing crossings in storyline visualizations.
\newblock In Emilio~Di Giacomo and Anna Lubiw, editors, {\em Proc. Graph Drawing and Network Visualization (GD'15)}, volume 9411 of {\em LNCS}, pages 192--198. Springer, 2015.
\newblock \href {https://doi.org/10.1007/978-3-319-27261-0_16} {\path{doi:10.1007/978-3-319-27261-0_16}}.

\bibitem{lwwll-stes-13}
Shixia Liu, Yingcai Wu, Enxun Wei, Mengchen Liu, and Yang Liu.
\newblock Storyflow: Tracking the evolution of stories.
\newblock {\em {IEEE} Trans. Vis. Comput. Graph.}, 19(12):2436--2445, 2013.
\newblock \href {https://doi.org/10.1109/TVCG.2013.196} {\path{doi:10.1109/TVCG.2013.196}}.

\bibitem{m-mnc-09}
Randall Munroe.
\newblock Movie narrative charts, 2009.
\newblock URL: \url{https://xkcd.com/657/}.

\bibitem{om-ses-10}
Michael Ogawa and Kwan{-}Liu Ma.
\newblock Software evolution storylines.
\newblock In Alexandru~C. Telea, Carsten G{\"{o}}rg, and Steven~P. Reiss, editors, {\em Software Visualization (SoftVis'10)}, pages 35--42. ACM, 2010.
\newblock \href {https://doi.org/10.1145/1879211.1879219} {\path{doi:10.1145/1879211.1879219}}.

\bibitem{p-wagehu-97}
Helen Purchase.
\newblock Which aesthetic has the greatest effect on human understanding?
\newblock In {\em Proc. Graph Drawing and Network Visualization (GD'97)}, volume 1353 of {\em LNCS}, pages 248--261. Springer, 1997.
\newblock \href {https://doi.org/10.1007/3-540-63938-1_67} {\path{doi:10.1007/3-540-63938-1_67}}.

\bibitem{s-cng-18}
Marcus Schaefer.
\newblock {\em Crossing Numbers of Graphs}.
\newblock CRC Press, 2018.

\bibitem{sbbzzm-mvnarmcc-18}
Yang Shi, Chris Bryan, Sridatt Bhamidipati, Ying Zhao, Yaoxue Zhang, and Kwan{-}Liu Ma.
\newblock Meetingvis: Visual narratives to assist in recalling meeting context and content.
\newblock {\em {IEEE} Trans. Vis. Comput. Graph.}, 24(6):1918--1929, 2018.
\newblock \href {https://doi.org/10.1109/TVCG.2018.2816203} {\path{doi:10.1109/TVCG.2018.2816203}}.

\bibitem{DBLP:journals/dm/Simone90}
Caterina~De Simone.
\newblock The cut polytope and the boolean quadric polytope.
\newblock {\em Discret. Math.}, 79(1):71--75, 1990.
\newblock \href {https://doi.org/10.1016/0012-365X(90)90056-N} {\path{doi:10.1016/0012-365X(90)90056-N}}.

\bibitem{DBLP:journals/tsmc/SugiyamaTT81}
Kozo Sugiyama, Shojiro Tagawa, and Mitsuhiko Toda.
\newblock Methods for visual understanding of hierarchical system structures.
\newblock {\em {IEEE} Trans. Syst. Man Cybern.}, 11(2):109--125, 1981.
\newblock \href {https://doi.org/10.1109/TSMC.1981.4308636} {\path{doi:10.1109/TSMC.1981.4308636}}.

\bibitem{MovieDataset2013}
Y.~Tanahashi.
\newblock Movie data set, 2013.
\newblock URL: \url{http://vis.cs.ucdavis.edu/~tanahashi/datadownloads/storyline visualizations/story_data.tar}.

\bibitem{thm-efgsvfsd-15}
Yuzuru Tanahashi, Chien{-}Hsin Hsueh, and Kwan{-}Liu Ma.
\newblock An efficient framework for generating storyline visualizations from streaming data.
\newblock {\em {IEEE} Trans. Vis. Comput. Graph.}, 21(6):730--742, 2015.
\newblock \href {https://doi.org/10.1109/TVCG.2015.2392771} {\path{doi:10.1109/TVCG.2015.2392771}}.

\bibitem{tm-dcosv-12}
Yuzuru Tanahashi and Kwan{-}Liu Ma.
\newblock Design considerations for optimizing storyline visualizations.
\newblock {\em {IEEE} Trans. Vis. Comput. Graph.}, 18(12):2679--2688, 2012.
\newblock \href {https://doi.org/10.1109/TVCG.2012.212} {\path{doi:10.1109/TVCG.2012.212}}.

\bibitem{tlwlkk-pcesvurl-20}
Tan Tang, Renzhong Li, Xinke Wu, Shuhan Liu, Johannes Knittel, Steffen Koch, Lingyun Yu, Peiran Ren, Thomas Ertl, and Yingcai Wu.
\newblock Plotthread: Creating expressive storyline visualizations using reinforcement learning.
\newblock {\em {IEEE} Trans. Vis. Comput. Graph.}, 27(2):294--303, 2021.
\newblock \href {https://doi.org/10.1109/TVCG.2020.3030467} {\path{doi:10.1109/TVCG.2020.3030467}}.

\bibitem{trlcyw-iechs-18}
Tan Tang, Sadia Rubab, Jiewen Lai, Weiwei Cui, Lingyun Yu, and Yingcai Wu.
\newblock {iStoryline}: Effective convergence to hand-drawn storylines.
\newblock {\em {IEEE} Trans. Vis. Comput. Graph.}, 25(1):769--778, 2019.
\newblock \href {https://doi.org/10.1109/TVCG.2018.2864899} {\path{doi:10.1109/TVCG.2018.2864899}}.

\bibitem{dfflmr-bcsv-17}
Thomas~C. van Dijk, Martin Fink, Norbert Fischer, Fabian Lipp, Peter Markfelder, Alexander Ravsky, Subhash Suri, and Alexander Wolff.
\newblock Block crossings in storyline visualizations.
\newblock {\em J. Graph Algorithms Appl.}, 21(5):873--913, 2017.
\newblock \href {https://doi.org/10.7155/JGAA.00443} {\path{doi:10.7155/JGAA.00443}}.

\bibitem{dlmw-csvwbc-17}
Thomas~C. van Dijk, Fabian Lipp, Peter Markfelder, and Alexander Wolff.
\newblock Computing storyline visualizations with few block crossings.
\newblock In Fabrizio Frati and Kwan{-}Liu Ma, editors, {\em Proc. Graph Drawing and Network Visualization (GD'17)}, volume 10692 of {\em LNCS}, pages 365--378. Springer, 2017.
\newblock \href {https://doi.org/10.1007/978-3-319-73915-1_29} {\path{doi:10.1007/978-3-319-73915-1_29}}.

\bibitem{wwd-vsegusvswll-23}
G{\"{u}}nter Wallner, Letian Wang, and Claire Dormann.
\newblock Visualizing the spatio-temporal evolution of gameplay using storyline visualization: {A} study with league of legends.
\newblock {\em Proc. {ACM} Hum. Comput. Interact.}, 7({CHI}):1002--1024, 2023.
\newblock \href {https://doi.org/10.1145/3611058} {\path{doi:10.1145/3611058}}.

\bibitem{wszlcl-evrwteed-24}
Yunchao Wang, Guodao Sun, Zihao Zhu, Tong Li, Ling Chen, and Ronghua Liang.
\newblock ${E}^2${S}toryline: Visualizing the relationship with triplet entities and event discovery.
\newblock {\em {ACM} Trans. Intell. Syst. Technol.}, 15(1):16:1--16:26, 2024.
\newblock \href {https://doi.org/10.1145/3633519} {\path{doi:10.1145/3633519}}.

\end{thebibliography}

\clearpage

\appendix

\section{Detailed statistics on instances and algorithm executions.}\label{apx:table}
\begin{center}
\begin{longtable}{|l|l|r|r|r|r|r|}
\caption{Table of runtimes for the 23 instances. B\&B = number of branch-and-bound nodes, Root LB = lower bound computed at the root node, number of constraints, overall runtime in seconds, and the crossing number or interval of best lower bound and upper bound if not solved within time limit. The value t.l. means time limit and $|\mathcal{C}|$, $\ell$, and $|\mathcal{I}|$ are the number of characters, layers, and interactions, respectively.} \label{table:runtimes}\\
\hline
Instance & Algorithm & B\&B & Root LB &  {\#}constraints & Time[$s$] & crossings\\\hline
\endfirsthead
\caption{Table of runtimes for the 23 instances (continued)}\\ %
\hline
Instance & Algorithm & B\&B & Root LB &  {\#}constraints & Time[$s$] & crossings\\\hline
\endhead
\textbf{Avatar} & \plo & 1 & 229 & $\num{238496}$ & 106.4 & 229 \\
$|\mathcal{C}|=\text{35}$ & \lin & 628 & 226 & $\num{820944}$ & 284.2 & 229 \\
$\ell=\text{54}$ & \qdr & 3295 & 80 &  & 358.5 & 229 \\
$|\mathcal{I}|=\text{54}$ & \mc & 1 & 229 &  & 841.9 & 229 \\\hline
\textbf{Barbie} & \plo & 1 & 138 & $\num{1125220}$ & 304.0 & 138 \\
$|\mathcal{C}|=\text{67}$ & \lin & 1 & 138 & $\num{6393964}$ & 123.4 & 138 \\
$\ell=\text{63}$ & \qdr & 558 & 48 &  & 66.6 & 138 \\
$|\mathcal{I}|=\text{63}$ & \mc & 1 & 138 &  & 923.6 & 138 \\\hline
\textbf{HarryPotter1} & \plo & 1 & 236 & $\num{520981}$ & 776.9 & 236 \\
$|\mathcal{C}|=\text{52}$ & \lin & 597 & 229 & $\num{2361796}$ & 1123.5 & 236 \\
$\ell=\text{54}$ & \qdr & 11288 & 77 &  & t.l. & [223,239] \\
$|\mathcal{I}|=\text{54}$ & \mc & 1 & 236 &  & 1444.9 & 236 \\\hline
\textbf{Oceans11} & \plo & 1 & 275 & $\num{1271525}$ & 536.9 & 275 \\
$|\mathcal{C}|=\text{60}$ & \lin & 578 & 271 & $\num{7018542}$ & 1087.9 & 275 \\
$\ell=\text{96}$ & \qdr & 11908 & 82 &  & t.l. & [244,284] \\
$|\mathcal{I}|=\text{96}$ & \mc & 1 & 275 &  & 2106.1 & 275 \\\hline
\textbf{anna1-2} & \plo & 1 & 57 & $\num{1370481}$ & 47.8 & 57 \\
$|\mathcal{C}|=\text{57}$ & \lin & 1 & 56 & $\num{6926150}$ & 38.2 & 57 \\
$\ell=\text{116}$ & \qdr & 561 & 38 &  & 37.6 & 57 \\
$|\mathcal{I}|=\text{116}$ & \mc & 1 & 57 &  & 53.1 & 57 \\\hline
\textbf{anna1-3} & \plo & 1 & 108 & $\num{4241743}$ & 1634.2 & 108 \\
$|\mathcal{C}|=\text{83}$ & \lin & 1 & 108 & $\num{30520530}$ & 966.5 & 108 \\
$\ell=\text{164}$ & \qdr & 1792 & 52 &  & t.l. & [104,113] \\
$|\mathcal{I}|=\text{164}$ & \mc & 1 & 108 &  & 2857.8 & 108 \\\hline
\textbf{anna2-3} & \plo & 1 & 28 & $\num{1793109}$ & 29.4 & 28 \\
$|\mathcal{C}|=\text{67}$ & \lin & 1 & 28 & $\num{10538684}$ & 12.8 & 28 \\
$\ell=\text{106}$ & \qdr & 100 & 22 &  & 9.3 & 28 \\
$|\mathcal{I}|=\text{106}$ & \mc & 1 & 28 &  & 30.6 & 28 \\\hline
\textbf{anna2-4} & \plo & 9 & 76 & $\num{3738699}$ & 701.4 & 76 \\
$|\mathcal{C}|=\text{80}$ & \lin & 1 & 76 & $\num{25966988}$ & 472.7 & 76 \\
$\ell=\text{155}$ & \qdr & 562 & 47 &  & 136.3 & 76 \\
$|\mathcal{I}|=\text{155}$ & \mc & 2 & 76 &  & 1339.5 & 76 \\\hline
\textbf{anna3-5} & \plo & 1 & 107 & $\num{4882063}$ & 1157.5 & 107 \\
$|\mathcal{C}|=\text{88}$ & \lin & 1 & 107 & $\num{36830405}$ & 919.4 & 107 \\
$\ell=\text{168}$ & \qdr & 2570 & 59 &  & 1168.5 & 107 \\
$|\mathcal{I}|=\text{168}$ & \mc & 1 & 107 &  & 1404.2 & 107 \\\hline
\pagebreak
\textbf{anna4-5} & \plo & 1 & 70 & $\num{1568650}$ & 80.9 & 70 \\
$|\mathcal{C}|=\text{60}$ & \lin & 1 & 70 & $\num{8222852}$ & 48.9 & 70 \\
$\ell=\text{120}$ & \qdr & 1642 & 37 &  & 124.8 & 70 \\
$|\mathcal{I}|=\text{120}$ & \mc & 1 & 70 &  & 112.1 & 70 \\\hline
\textbf{anna4-6} & \plo & 1 & 157 & $\num{3228069}$ & 2087.8 & 157 \\
$|\mathcal{C}|=\text{71}$ & \lin & 119 & 157 & $\num{19956977}$ & 1960.4 & 157 \\
$\ell=\text{176}$ & \qdr & 1090 & 77 &  & 2138.1 & 157 \\
$|\mathcal{I}|=\text{176}$ & \mc & 1 & 157 &  & 2890.4 & 157 \\\hline
\textbf{anna5-6} & \plo & 1 & 76 & $\num{1821015}$ & 84.6 & 76 \\
$|\mathcal{C}|=\text{63}$ & \lin & 1 & 76 & $\num{10021591}$ & 64.9 & 76 \\
$\ell=\text{127}$ & \qdr & 782 & 37 &  & 92.6 & 76 \\
$|\mathcal{I}|=\text{127}$ & \mc & 1 & 76 &  & 83.0 & 76 \\\hline
\textbf{anna6-7} & \plo & 1 & 78 & $\num{1586302}$ & 155.5 & 78 \\
$|\mathcal{C}|=\text{60}$ & \lin & 1 & 78 & $\num{8300109}$ & 73.9 & 78 \\
$\ell=\text{118}$ & \qdr & 917 & 33 &  & 110.2 & 78 \\
$|\mathcal{I}|=\text{118}$ & \mc & 1 & 78 &  & 232.6 & 78 \\\hline
\textbf{anna6-8} & \plo & 193 & 121 & $\num{2266237}$ & 2700.7 & 121 \\
$|\mathcal{C}|=\text{65}$ & \lin & 556 & 118 & $\num{13025530}$ & 2155.1 & 121 \\
$\ell=\text{146}$ & \qdr & 3193 & 52 &  & 2352.8 & 121 \\
$|\mathcal{I}|=\text{146}$ & \mc & 1 & 121 &  & 2280.4 & 121 \\\hline
\textbf{anna7} & \plo & 1 & 8 & $\num{525061}$ & 0.6 & 9 \\
$|\mathcal{C}|=\text{47}$ & \lin & 1 & 8 & $\num{2151380}$ & 0.6 & 9 \\
$\ell=\text{62}$ & \qdr & 3 & 6 &  & 0.4 & 9 \\
$|\mathcal{I}|=\text{62}$ & \mc & 824 & 8 &  & 39.2 & 9 \\\hline
\textbf{gdea20} & \plo & 1 & 41 & $\num{108934}$ & 26.5 & 41 \\
$|\mathcal{C}|=\text{19}$ & \lin & 1 & 41 & $\num{212494}$ & 11.3 & 41 \\
$\ell=\text{100}$ & \qdr & 171 & 19 &  & 8.6 & 41 \\
$|\mathcal{I}|=\text{100}$ & \mc & 1 & 41 &  & 92.2 & 41 \\\hline
\textbf{huck} & \plo & 1 & 41 & $\num{2130471}$ & 5.7 & 42 \\
$|\mathcal{C}|=\text{74}$ & \lin & 1 & 41 & $\num{13537255}$ & 3.3 & 42 \\
$\ell=\text{107}$ & \qdr & 140 & 30 &  & 4.5 & 42 \\
$|\mathcal{I}|=\text{107}$ & \mc & 9547 & 41 &  & 2148.8 & 42 \\\hline
\textbf{jean1-3} & \plo & 1 & 53 & $\num{3616478}$ & 72.8 & 53 \\
$|\mathcal{C}|=\text{73}$ & \lin & 1 & 53 & $\num{24790611}$ & 44.4 & 53 \\
$\ell=\text{253}$ & \qdr & 1070 & 36 &  & 108.9 & 53 \\
$|\mathcal{I}|=\text{253}$ & \mc & 1 & 53 &  & 123.5 & 53 \\\hline
\textbf{jean1-4} & \plo & 1 & 172 & $\num{5643798}$ & 3306.5 & 172 \\
$|\mathcal{C}|=\text{79}$ & \lin & 1 & 172 & $\num{42117296}$ & 2739.6 & 172 \\
$\ell=\text{329}$ & \qdr & 1792 & 75 &  & t.l. & [149,183] \\
$|\mathcal{I}|=\text{329}$ & \mc & 1 & 156 &  & t.l. & [156,565] \\\hline
\textbf{jean2-3} & \plo & 1 & 33 & $\num{663871}$ & 30.7 & 33 \\
$|\mathcal{C}|=\text{41}$ & \lin & 1 & 33 & $\num{2671943}$ & 11.0 & 33 \\
$\ell=\text{158}$ & \qdr & 181 & 24 &  & 7.5 & 33 \\
$|\mathcal{I}|=\text{158}$ & \mc & 1 & 33 &  & 56.3 & 33 \\\hline
\pagebreak
\textbf{jean2-4} & \plo & 1 & 154 & $\num{1358514}$ & 3150.6 & 154 \\
$|\mathcal{C}|=\text{47}$ & \lin & 210 & 154 & $\num{6283384}$ & t.l. & [154,159] \\
$\ell=\text{234}$ & \qdr & 1587 & 68 &  & t.l. & [138,156] \\
$|\mathcal{I}|=\text{234}$ & \mc & 1 & 148 &  & t.l. & [148,468] \\\hline
\textbf{jean3-4} & \plo & 1 & 128 & $\num{780027}$ & 1943.7 & 128 \\
$|\mathcal{C}|=\text{41}$ & \lin & 330 & 126 & $\num{3209062}$ & t.l. & [127,132] \\
$\ell=\text{175}$ & \qdr & 1720 & 53 &  & 2157.5 & 128 \\
$|\mathcal{I}|=\text{175}$ & \mc & 1 & 127 &  & t.l. & [127,170] \\\hline
\textbf{jean4-5} & \plo & 1 & 96 & $\num{521469}$ & 125.7 & 96 \\
$|\mathcal{C}|=\text{36}$ & \lin & 1 & 96 & $\num{1939809}$ & 111.3 & 96 \\
$\ell=\text{149}$ & \qdr & 1476 & 50 &  & 171.2 & 96 \\
$|\mathcal{I}|=\text{149}$ & \mc & 1 & 96 &  & 238.0 & 96 \\\hline
\end{longtable}
\end{center}

\end{document}